\documentclass[sigconf, acm, screen]{acmart}

\newtheorem{myExa}{Example}
\newtheorem{myLemma}{Lemma}
\newtheorem{myTheorem}{Theorem}

\usepackage[ruled,vlined]{algorithm2e}
\usepackage{algpseudocode}
\usepackage{multirow}
\usepackage{stackengine}
\usepackage{makecell}
\usepackage{setspace}

\newcommand{\method}[1]{#1}
\newcommand{\methodbt}[1]{#1}
\newcommand{\squishlist}{
	\begin{list}{$\bullet$}
		{ \setlength{\itemsep}{1pt}
			\setlength{\parsep}{1pt}
			\setlength{\topsep}{2.5pt}
			\setlength{\partopsep}{0.5pt}
			\setlength{\leftmargin}{1em}
			\setlength{\labelwidth}{1em}
			\setlength{\labelsep}{0.6em}
		}
	}
	\newcommand{\squishend}{
	\end{list}
}

\usepackage{tikz}

\makeatletter
  \newcommand\figcaption{\def\@captype{figure}\caption}
  \newcommand\tabcaption{\def\@captype{table}\caption}
\makeatother

\makeatletter
\patchcmd{\@algocf@start}
  {-1.5em}
  {0pt}
  {}{}
\makeatother

\setcopyright{acmcopyright}
\copyrightyear{2025}
\acmYear{2025}
\acmDOI{XXXXXXX.XXXXXXX}

\acmConference[SIGMOD '25]{Make sure to enter the correct
  conference title from your rights confirmation emai}{June 22--27,
  2025}{Berlin, Germany}

\begin{document}

\title{Accelerating Graph Indexing for ANNS on Modern CPUs}

\author{Mengzhao Wang}
\affiliation{%
  \institution{Zhejiang University}
}
\email{wmzssy@zju.edu.cn}

\author{Haotian Wu}
\affiliation{%
  \institution{Zhejiang University}
}
\email{haotian.wu@zju.edu.cn}

\author{Xiangyu Ke}
\affiliation{%
  \institution{Zhejiang University}
}
\email{xiangyu.ke@zju.edu.cn}

\author{Yunjun Gao}
\affiliation{%
  \institution{Zhejiang University}
}
\email{gaoyj@zju.edu.cn}

\author{Yifan Zhu}
\affiliation{%
  \institution{Zhejiang University}
}
\email{xtf_z@zju.edu.cn}

\author{Wenchao Zhou}
\affiliation{%
  \institution{Alibaba Group}
}
\email{zwc231487@alibaba-inc.com}

\begin{abstract}
{In high-dimensional vector spaces, Approximate Nearest Neighbor Search (ANNS) is a key component in database and artificial intelligence infrastructures. 
Graph-based methods, particularly HNSW, have emerged as leading solutions among various ANNS approaches, offering an impressive trade-off between search efficiency and accuracy. 
Many modern vector databases utilize graph indexes as their core algorithms, benefiting from various optimizations to enhance search performance. However, the high indexing time associated with graph algorithms poses a significant challenge, especially given the increasing volume of data, query processing complexity, and dynamic index maintenance demand. This has rendered indexing time a critical performance metric for users.}

{In this paper, we comprehensively analyze the underlying causes of the low graph indexing efficiency on modern CPUs, identifying that distance computation dominates indexing time, primarily due to high memory access latency and suboptimal arithmetic operation efficiency. 
We demonstrate that distance comparisons during index construction can be effectively performed using compact vector codes at an appropriate compression error. Drawing from insights gained through integrating existing compact coding methods in the graph indexing process, we propose a novel compact coding strategy, named \texttt{Flash}, designed explicitly for graph indexing and optimized for modern CPU architectures. By minimizing random memory accesses and maximizing the utilization of SIMD (Single Instruction, Multiple Data) instructions, \texttt{Flash} significantly enhances cache hit rates and arithmetic operations. Extensive experiments conducted on eight real-world datasets, ranging from ten million to one billion vectors, exhibit that \texttt{Flash} achieves a speedup of 10.4$\times$ to 22.9$\times$ in index construction efficiency, while maintaining or improving search performance.}

\end{abstract}

\maketitle

\section{Introduction}
\label{sec: intro}
Approximate Nearest Neighbor Search (ANNS) in high-dimensional spaces has become integral due to the advent of deep learning embedding techniques for managing unstructured data \cite{Milvus_sigmod2021,graph_survey_vldb2021,tau-MG,RaBitQ}. This spans a variety of applications, such as recommendation systems \cite{ChenLZWLMHJXDZ22,HanZYXCS23}, information retrieval \cite{GouDWK24,WangMW22}, and vector databases \cite{VDB_tutorial_SIGMOD24,Manu_zilliz}. 
Owing to the ``curse of dimensionality'' and the ever-increasing volume of data, ANNS seeks to optimize both speed and accuracy, making it more practical compared to the exact solutions, which can be prohibitively time-consuming.
Given a query vector, ANNS efficiently finds its $k$ nearest vectors in a vector dataset following a specific distance metric. 
Current studies explore four index types for ANNS: Tree-based \cite{AroraSK018,LuWWK20}, Hash-based \cite{HuangFZFN15,ZhaoZYLXZJ23}, Quantization-based \cite{PQ,PQfast}, and Graph-based methods \cite{NSG,HVS}. 
Among these, graph-based algorithms like \method{HNSW} \cite{HNSW} have demonstrated superior performance, establishing themselves as the industry standard \cite{lanns,es_hnsw,douze2024faiss} and attracting substantial academic interest \cite{ADSampling,HM_ANN,LiZAH20}.
For example, \method{HNSW} underpins various ANNS services or vector databases, including PostgreSQL \cite{PASE}, ElasticSearch \cite{es_hnsw}, Milvus \cite{Milvus_sigmod2021}, and Pinecone \cite{pinecone_hnsw}. 
Recent publications from the data management community \cite{sigmod_24_papers, vldb_24_papers, icde_24_papers} highlight the focus on ANNS research regarding \method{HNSW}\footnote{Unless otherwise specified, we illustrate our framework design using this representative graph-based index.}.

Although numerous studies \cite{LiZAH20,ADSampling,HVS,yue2023routing,Finger,ColemanSSS22,PengZLJR23} have explored optimizations for search performance, research on improving index construction efficiency remains limited\footnote{{Since multi-thread graph indexing on CPUs is a standard practice, we present all discussed graph algorithms in their multi-thread versions by default.}}, which is increasingly critical in real-world applications \cite{XuLLXCZLYYYCY23}. 
First, the {\em vast volume of data} poses significant challenges for index construction. For instance, building an \method{HNSW} index on tens of millions of vectors typically requires about \textit{10 hours} \cite{HVS,lanns}, billion-scale datasets may extend construction time to \textit{five days}, even with relaxed construction parameters \cite{DiskANN}. 
Second, complex query processing, such as hybrid search \cite{Filtered-DiskANN,mengzhao_neurips2023}, adds constraints on ANNS, {\em complicating the construction process} and increasing indexing times \cite{ZuoQZLD24}. For example, constructing a specialized \method{HNSW} index for attribute-constrained ANNS takes 33$\times$ longer than a standard index \cite{PatelKGZ24}. 
More importantly, {\em continuous data and embedding model updates} \cite{Neos,Jingdong_paper,SundaramTSMIMD13,embedding_models_update_1,embedding_models_update_2} necessitate a periodic reconstruction process based on the LSM-Tree framework \cite{Milvus_sigmod2021,ADBV,pan2023survey,nmslib_issues73,SingleStore-V,XuLLXCZLYYYCY23}.
While some approaches aim to avoid rebuilding \cite{Fresh-DiskANN}, they are often unsuitable for embedding model updates \cite{PouyanfarSYTTRS19,embedding_models_update_1} and can result in diminished search performance \cite{nmslib_issues73}, with accuracy dropping from 0.95 to 0.88 after 20 update cycles on \method{HNSW} \cite{Fresh-DiskANN}. 
Thus, reconstruction efficiency has become a bottleneck in modern vector databases \cite{pan2023survey}, as it {\bf \em directly impacts the ability to rapidly release new services with up-to-date data and high search performance} \cite{Fresh-DiskANN,ADBV}. In business scenarios, index rebuilding often occurs overnight during low user activity \cite{ADBV}, constrained to a few hours. Serving approximately 100M vectors on a single node, the \method{HNSW} build time frequently exceeds 12 hours, failing to meet the requirement.

To enhance the index construction efficiency, employing specialized hardware like GPUs is a straightforward approach \cite{WangZZY21,ShiZZJHLH18}. 
By leveraging GPUs' parallel computing capabilities, recent studies have parallelized distance computations and data structure maintenance for the graph indexing process on \textit{million-scale} datasets \cite{GANNS,CAGRA,SONG}, achieving an order of magnitude speedup over single-thread CPU implementations \cite{GANNS}. 
However, challenges such as high costs, memory constraints, and additional engineering efforts limit their practical applications \cite{ZhangL0L024}. 
For instance, even high-end GPUs like the NVIDIA A100, with tens of gigabytes of memory, struggle to accommodate large vector datasets requiring hundreds of gigabytes. 
As noted by Stonebraker in his recent survey \cite{Stonebraker_SIGMOD_Record}, ``\textit{If data does not fit in GPU memory, query execution bottlenecks on loading data into the device, significantly diminishing parallelization benefits.}'' Consequently, CPU-based deployment remains the prevalent choice in practical scenarios \cite{Milvus_sigmod2021,ColemanSSS22,douze2024faiss,PASE}, providing a cost-effective, memory-sufficient, and easy-to-use solution. 
This motivates us to explore optimization opportunities to improve the index construction efficiency on modern CPUs.

\begin{figure}
  \setlength{\abovecaptionskip}{0cm}
  \setlength{\belowcaptionskip}{-0.4cm}
  \centering
  \footnotesize
  \stackunder[0.5pt]{\includegraphics[scale=0.33]{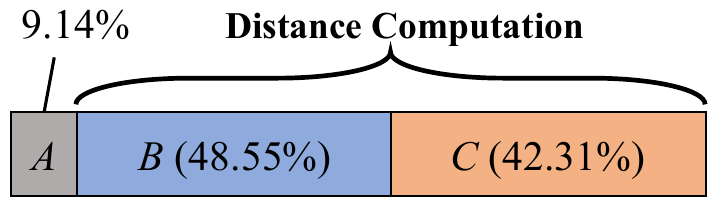}}{(a) LAION-1M ($D=768$)}
  \hspace{0.15cm}
  \stackunder[0.5pt]{\includegraphics[scale=0.32]{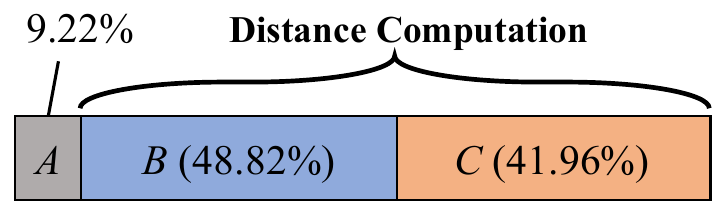}}{(b) ARGILLA-1M ($D=1024$)}
  \newline
  \caption{Profiling of HNSW indexing time. Distance computation constitutes the majority of the indexing time and consists of two components: memory accesses (\textit{B}) and arithmetic operations (\textit{C}). Other tasks (\textit{A}), such as data structure maintenance, account for a minor portion.}
  \label{fig: indexing profile}
  \vspace{-0.4cm}
\end{figure}

In this paper, we conduct an in-depth analysis of the \method{HNSW} construction process, identifying the root causes of inefficiency on modern CPUs. Figure \ref{fig: indexing profile} illustrates HNSW indexing time profiled using the \textit{perf} tool \cite{perf} on two real-world embedding datasets \cite{laion, argilla}. Results indicate that distance computation constitutes over 90\% of the total construction time, marking it as the major bottleneck, while memory accesses and arithmetic operations contribute similarly.
The lack of spatial locality in the \method{HNSW} index necessitates random memory accesses to fetch vectors for almost every distance computation. Building \method{HNSW} on a dataset of size $n$ requires $O(n\log(n))$ distance calculations \cite{HNSW}, leading to frequent cache misses as target data is often uncached. As a result, the memory controller frequently loads data from main memory, resulting in high access latency.
Additionally, current vector data typically consists of high-dimensional floating-point numbers, occupying $4\cdot D$ bytes for dimension $D$ (e.g., 3KB for $D=768$). However, Single Instruction, Multiple Data (SIMD) registers are restricted to hundreds of bits (e.g., 128 bits for the SSE instruction set) \cite{intel_simd}, much smaller than vector sizes. Thus, each distance computation requires hundreds of memory reads to load vector segments into a register, reducing the efficiency of SIMD operations. Notably, all distances are calculated solely for comparison during HNSW construction (refer to Section \textbf{\ref{subsec: problem analysis}}). Recent studies suggest that exact distances are unnecessary for comparison \cite{ADSampling,yang2024bridging,wang2024distance}.

To address these issues, we optimize the \method{HNSW} construction process by reducing random memory accesses and effectively utilizing SIMD instructions. The theoretical analysis demonstrates that distance comparisons during index construction can be performed effectively using compact vector codes at an appropriate compression error. We initially implement mainstream vector compression methods—Product Quantization (PQ) \cite{PQ}, Scalar Quantization (SQ) \cite{LVQ}, and Principal Component Analysis (PCA) \cite{PCA}—into the \method{HNSW} construction process. However, these methods yield limited improvements in indexing efficiency or degrade search performance, as they do not align with HNSW's construction characteristics. Further observations show that they fail to reduce excessive random memory accesses and do not leverage SIMD advantages.
This motivates us to develop a speci\underline{\textbf{f}}ic compact coding strategy, a\underline{\textbf{l}}ongside memory l\underline{\textbf{a}}yout optimization for the con\underline{\textbf{s}}truction process of \underline{\textbf{H}}NSW, named \texttt{Flash}. To maximize SIMD instruction utilization, \texttt{Flash} is designed to accommodate SIMD register features, enabling concurrent execution of distance computations within the SIMD register. To minimize random memory accesses, \texttt{Flash} is effectively organized to match data access and register load patterns.
Ultimately, \texttt{Flash} avoids numerous random memory accesses and enhances SIMD usage, achieving an order of magnitude speedup in construction efficiency while maintaining or improving search performance. To the best of our knowledge, this is the first work to reach such significant speedup on modern CPUs.

To sum up, this paper makes the following contributions.

\squishlist

\item We deeply analyze the construction process of \method{HNSW} and identify the root causes of low construction efficiency on modern CPUs. Specifically, we find that distance computation suffers from high memory access latency and low arithmetic operation efficiency, which accounts for over 90\% of indexing time.

\item We propose a novel compact coding strategy, \texttt{Flash}, specifically designed for \method{HNSW} construction, optimizing index layout for efficient memory accesses and SIMD utilization. This approach improves cache hits and register-resident computation.

\item We conduct extensive evaluations on real-world datasets, ranging from ten million to billion-scale, demonstrating that construction efficiency can be improved by over an order of magnitude, while search performance remains the same or even improves.

\item We summarize three notable findings from our research: (1) a compact coding method that significantly enhances search performance may not be suitable for index construction; (2) reducing more dimensions may bring higher accuracy; and (3) encoding vectors and distances with a tiny amount of bits to align with hardware constraints may yield substantial benefits.

\squishend

The paper is organized as follows: Section \ref{sec: background} reviews related work and conducts a problem study on index construction of HNSW. Section \ref{sec: data processing} introduces the proposed novel compact coding strategy alongside the memory layout optimization for \method{HNSW} construction. Section \ref{sec: experiments} provides experimental results, while some notable findings and conclusion are outlined in Sections \ref{sec: summary} and \ref{sec: conclusion}, respectively.
\section{Background and Motivation}
\label{sec: background}
In this section, we first review current graph-based ANNS algorithms to ascertain the leadership of \method{HNSW}, then outline search and indexing optimizations based on \method{HNSW}. Next, we provide a detailed analysis of the \method{HNSW} construction process to identify the root causes of inefficiency on modern CPUs.

\subsection{Related Work}
\label{subsec: Related Work}

\subsubsection{\textbf{Graph-based ANNS Algorithms}}
\label{subsubsec: Graph ANNS Alg}
Graph-based ANNS algorithms have emerged as the leading approach for balancing search accuracy and efficiency in high-dimensional spaces \cite{graph_survey_vldb2021,DPG,DiskANN,tau-MG,NSSG,RoarGraph}. These methods construct a graph index where vertices represent database vectors, and directed edges signify neighbor relationships. During search, a greedy algorithm navigates the graph, visiting a vertex's neighbors and selecting the one closest to the query, continuing until no closer neighbor is found, yielding the final result.
Current methods share similar index structures and greedy search strategies but differ in their edge selection strategies during graph construction \cite{graph_survey_vldb2021}. For example, \method{HNSW} \cite{HNSW} and \method{NSG} \cite{NSG} use the Monotonic Relative Neighborhood Graph (MRNG) strategy, while \method{Vamana} \cite{DiskANN} and $\tau$-\method{MG} \cite{tau-MG} incorporate additional parameters. \method{HCNNG} \cite{HCNNG} employs the Minimum Spanning Tree (MST) approach, while \method{KGraph} \cite{NNDescent} and NGT \cite{NGT} implement the K-Nearest Neighbor Graph (KNNG) paradigm. {Despite these variations, state-of-the-art methods like \method{HNSW}, \method{NSG}, Vamana, and $\tau$-\method{MG} adhere to a similar construction framework, consisting of Candidate Acquisition (CA) and Neighbor Selection (NS) stages \cite{NSG,tau-MG}.} {Our research identifies the bottleneck in the CA and NS steps, where distance calculations dominate indexing time (see Section \ref{subsec: problem analysis}). Consequently, improving the CA and NS steps will accelerate the indexing process in all these graph algorithms.} Notably, \method{HNSW} stands out for its robust optimizations and versatile features \cite{Finger,ZuoQZLD24,LiuZHSLLDYW22}, including native \textit{add} support. Numerous studies have enhanced \method{HNSW} from various perspectives (we review them in the following paragraphs), solidifying its role as a mainstream method in both industry and academia.

\subsubsection{\textbf{Search Optimizations on \methodbt{HNSW}}}
\label{subsubsec: Search Opt HNSW}
Search optimizations for \method{HNSW} focus on several key components. Some aim to optimize entry point acquisition by leveraging additional quantization or hashing structures to start closer to the query \cite{HVS,zhao2023towards}. Others optimize neighbor visits by learning global \cite{LearnToRoute} or topological \cite{MunozDT19} information, enhancing data locality \cite{ColemanSSS22}, and partitioning neighbor areas \cite{TOGG}. These strategies improve accuracy and efficiency by identifying relevant neighbors and skipping irrelevant ones. Some approaches reduce distance computation complexity using approximate calculations \cite{ADSampling,Finger,yang2024bridging,lu2024probabilistic}, enhancing search efficiency with minimal accuracy loss. Additionally, optimizations refine termination conditions through adaptive decision-making \cite{LiZAH20,LiuZHSLLDYW22,yang2021tao} and relaxed monotonicity \cite{zhang2023vbase}.
Specialized hardware like GPUs \cite{SONG}, FPGAs \cite{jiang2024accelerating,ZengZLZDZLNXYW23,PengCWYWGLBSJLD21}, Compute Express Link (CXL) \cite{CXL-ANNS}, Non-Volatile Memory (NVM) \cite{HM_ANN}, NVMe SSDs \cite{Starling}, and SmartSSDs \cite{SmartSSD} accelerate distance computation and optimize the \method{HNSW} index layout, achieving superior performance through software-hardware collaboration. To support complex query processing, some approaches integrate structured predicates \cite{PatelKGZ24,NHQ} or range attributes \cite{ZuoQZLD24} into \method{HNSW}, enabling hybrid searches of vectors and attributes. These optimizations enhance HNSW's performance and versatility, meeting diverse real-world requirements, though most increase indexing time compared to the original \method{HNSW}.

\subsubsection{\textbf{Indexing Optimizations of \methodbt{HNSW}}}
\label{subsubsec: Index Opt HNSW}
Despite HNSW's superior search performance, its low index construction efficiency remains a major bottleneck, a challenge shared by other graph-based methods \cite{HVS,NSG,NSSG,DiskANN,zhao2023towards,aumuller2023recent}. As demands for dynamic updates \cite{XuLLXCZLYYYCY23,Fresh-DiskANN}, large-scale data \cite{DiskANN,SPANN}, and complex queries \cite{PatelKGZ24,ZuoQZLD24} grow (see Section \ref{sec: intro}), optimizing HNSW construction becomes increasingly important. Current research targets two areas: algorithm optimization and hardware acceleration.
Algorithm optimization efforts have redesigned the construction pipeline, either fully \cite{NSG,tau-MG} or partially \cite{HVS}, but these attempts offer only marginal improvements \cite{NSG,NSSG,zhao2023towards}. Moreover, these optimizations substantially modify the algorithmic process of \method{HNSW}, requiring code refactoring that may disrupt many engineering optimizations integrated over the years \cite{hnswlib,n2,Faiss}. Some HNSW features, like native \textit{add} support, are weakened \cite{HVS} or even discarded \cite{NSG}.
On the other hand, GPU-accelerated methods parallelize distance computations during \method{HNSW} construction \cite{WangZZY21,CAGRA,GGNN}. These methods re-implement \method{HNSW} to align with GPU architecture, achieving an order of magnitude speedup compared to single-thread CPU-based construction \cite{GANNS}. However, GPU-specific limitations reduce their practicality in real-world scenarios \cite{ZhangL0L024} (see Section \ref{sec: intro}). As CPU-based \method{HNSW} remains prevalent in many industrial applications \cite{Milvus_sigmod2021,ColemanSSS22,douze2024faiss,PASE}, accelerating \method{HNSW} construction on modern CPUs is crucial. {Note that parallel multi-thread implementations of graph indexing on CPUs are simple, efficient, and widely adopted by current graph indexes (e.g., \method{HNSW}, \method{NSG} \cite{NSG}, $\tau$-MG \cite{tau-MG}). These implementations are orthogonal to our work, as our techniques build upon multi-thread graph indexing.}

\subsubsection{\textbf{Distributed Deployment of \methodbt{HNSW}}}
\label{subsubsec: Distri Deploy HNSW}
Distributed deployment accelerates HNSW construction by partitioning large datasets into multiple shards, each containing tens of millions of vectors \cite{Manu_zilliz,lanns,IwabuchiSPPS23,deng2019pyramid}. This enables concurrent index construction across nodes, enhancing efficiency for large-scale datasets \cite{lanns}. Query processing can be optimized by an inter-shard coordinator \cite{Manu_zilliz}, and advanced data fragmentation strategies can assign queries to relevant shards \cite{ZhangYGWHLLZT23,ELPIS}. At the system level, distributed deployment ensures scaling, load balancing, and fault tolerance \cite{Starling}, but introduces challenges like high communication overhead \cite{MiaoZSNYTC21,Milvus_sigmod2021}. Notably, distributed deployment is orthogonal to our research, as we focus on efficient index building within a shard, which can be directly integrated into existing distributed systems.

\subsection{Problem Analysis}
\label{subsec: problem analysis}

\setlength{\textfloatsep}{0pt}
\begin{algorithm}[t]
\label{alg: hnsw construction}
  \caption{\textsc{Index Construction of \method{HNSW}}}
  \LinesNumbered
  \KwIn{a vector dataset $\boldsymbol{S}$, hyper-parameters $C$ and $R$ ($R\leq C$)}
  \KwOut{\method{HNSW} index built on $\boldsymbol{S}$}

  $\boldsymbol{V} \gets \emptyset$, $\boldsymbol{E} \gets \emptyset$; \Comment{\textsf{start from a empty graph}}

  \For(\Comment{\textsf{insert all vectors in $\boldsymbol{S}$}}){each $\boldsymbol{x}$ in $\boldsymbol{S}$}{
    $l_{max}$ $\gets$ $\boldsymbol{x}$'s max layer; \Comment{\textsf{exponential decaying distribution}}

    \For(\Comment{\textsf{insert $\boldsymbol{x}$ at layer $l$}}){each $l$ $\in$ $\{l_{max}, \cdots, 0\}$}{
      $\boldsymbol{C(x)}$ $\gets$ top-$C$ candidates; \Comment{\textsf{greedy search}}

      $\boldsymbol{N(x)}$ $\gets$ $\leq R$ neighbors from $\boldsymbol{C(x)}$; \Comment{\textsf{heuristic strategy}}

      add $\boldsymbol{x}$ to $\boldsymbol{N(y)}$ for each $\boldsymbol{y}\in \boldsymbol{N(x)}$; \Comment{\textsf{reverse edge}}
    }

    $\boldsymbol{V}\gets \boldsymbol{V}\cup \{\boldsymbol{x}\}$, $\boldsymbol{E}\gets \boldsymbol{E} \cup \boldsymbol{N(x)}$
    
  }
  
  \textbf{return} $\boldsymbol{G}=(\boldsymbol{V},\boldsymbol{E})$
\end{algorithm}

\noindent\underline{Notations.} In the HNSW index, each vertex corresponds to a unique vector, denoted by bold lowercase letters (e.g., $\boldsymbol{x}$, $\boldsymbol{y}$). We use the Euclidean distance $\delta(\boldsymbol{x}, \boldsymbol{y})$ as the distance metric for quantifying similarity, where a smaller $\delta(\boldsymbol{x}, \boldsymbol{y})$ indicates greater similarity \cite{NSG,graph_survey_vldb2021,zhao2023towards,tau-MG}. Bold uppercase letters (e.g., $\boldsymbol{S}$) denote sets, while unbolded uppercase and lowercase letters represent parameters.

Given a vector dataset $\boldsymbol{S}$, the \method{HNSW} index $\boldsymbol{G}$ is built by progressively inserting vectors into the current graph index, as outlined in Algorithm \hyperref[alg: hnsw construction]{1}. Two hyper-parameters must be set before index construction: the maximum number of candidates ($C$) and the maximum number of neighbors ($R$)\footnote{In the original paper and many popular repositories \cite{n2,hnswlib,HNSW}, $C$ is called $efConstruction$, and $R$ in the base layer is twice that in higher layers \cite{HNSW}.}. For each inserted vector $\boldsymbol{x}$, the maximum layer $l_{max}$ is randomly selected following an exponentially decaying distribution (line 3). The vector is inserted into layers from $l_{max}$ to $0$, with the same procedure applied at each layer (lines 4-7). At layer $l$, $\boldsymbol{x}$ is treated as a query point, and a greedy search on the current graph at layer $l$ identifies the top-$C$ nearest vertices as candidates (line 5). During the search, a candidate set $\boldsymbol{C(x)}$ of size $C$ is maintained, with vertices ordered in ascending distance to $\boldsymbol{x}$, and the maximum distance is denoted as $T$. The search iteratively visits a vertex's neighbors, calculates their distances to $\boldsymbol{x}$, and updates $\boldsymbol{C(x)}$ with neighbors that have smaller distances ($<T$). To obtain the final neighbors, a heuristic edge selection strategy prevents clustering among neighbors (line 6). For example, for candidate $\boldsymbol{v}$ and any candidate $\boldsymbol{u}$ where $\delta(\boldsymbol{u},\boldsymbol{x}) < \delta(\boldsymbol{v},\boldsymbol{x})$, if $\delta(\boldsymbol{u},\boldsymbol{v}) < \delta(\boldsymbol{v},\boldsymbol{x})$, $\boldsymbol{v}$ is excluded from $\boldsymbol{N(x)}$. After determining $\boldsymbol{x}$'s final neighbors, $\boldsymbol{x}$ is added to the neighbor list $\boldsymbol{N(y)}$ for each $\boldsymbol{y} \in \boldsymbol{N(x)}$ (line 7). If $\boldsymbol{N(y)}$ exceeds $R$, the heuristic edge selection strategy prunes excess neighbors, ensuring that the neighbor count remains under $R$.

The construction process of \method{HNSW} involves two main steps: \textit{Candidate Acquisition} (CA) and \textit{Neighbor Selection} (NS). The CA step identifies the top-$C$ nearest vertices for an inserted vector using a greedy search strategy. This involves visiting a vertex's neighbor list and computing distances to the inserted vector, requiring random memory access to fetch a neighbor's vector data, stored separately from neighbor IDs due to the large vector size. In the NS step, the final $R$ neighbors are selected from the candidate set using a heuristic strategy. Here, distances among candidates are computed, also requiring random access to their vector data. In both CA and NS, distance calculations are used solely for comparison. In CA, distances are compared to the largest distance in the candidate set to determine if a visiting vertex can replace the farthest candidate. In NS, distances among candidates and between candidates and the inserted vector are compared to decide the inclusion of a candidate in the final neighbor set. HNSW currently performs all distance computations on full-precision vectors. Additionally, to utilize SIMD acceleration, a distance computation requires hundreds of read operations to load vector segments into a 128-bit register due to a large vector size, often spanning thousands of bytes.

\begin{figure}
  \centering
  \setlength{\abovecaptionskip}{0.1cm}
  \setlength{\belowcaptionskip}{0.1cm}
  \includegraphics[width=\linewidth]{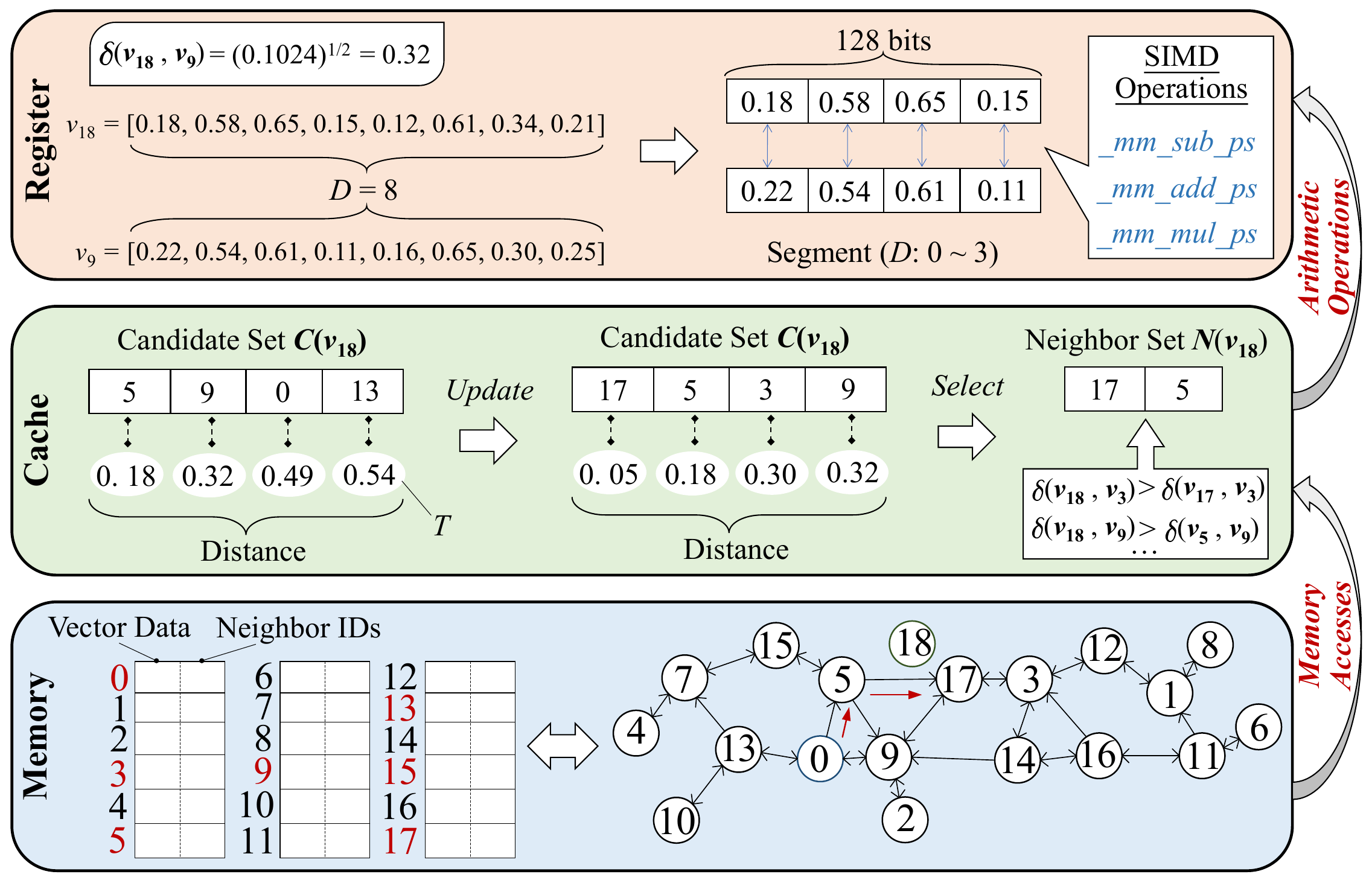}
  \caption{Illustrating memory accesses and arithmetic operations during the \method{HNSW} construction process (base layer).}
  \label{fig:problem ananlysis}
\end{figure}

\begin{myExa}
    \rm In Figure \ref{fig:problem ananlysis}, we illustrate the index construction process of \method{HNSW} at the base layer, exemplified by inserting $\boldsymbol{v_{18}}$ with $C=4$ and $R=2$. In the CA step, $\boldsymbol{v_{18}}$ is treated as a query point, and a greedy search from $\boldsymbol{v_0}$ is initiated. After examining $\boldsymbol{v_0}$'s neighbors, the candidate set $\boldsymbol{C(v_{18})}$ is $\{\boldsymbol{v_5}, \boldsymbol{v_9}, \boldsymbol{v_0}, \boldsymbol{v_{13}} \}$, ordered by proximity to $\boldsymbol{v_{18}}$. The current distance threshold $T$ is $\delta(\boldsymbol{v_{13}}, \boldsymbol{v_{18}})=0.54$. The nearest vertex, $\boldsymbol{v_5}$, is then visited, prompting an exploration of its neighbors to update $\boldsymbol{C(v_{18})}$. Upon calculating $\delta(\boldsymbol{v_{15}}, \boldsymbol{v_{18}})$ and comparing it with $T$, $\boldsymbol{v_{13}}$ is replaced by $\boldsymbol{v_{15}}$ as $\delta( \boldsymbol{v_{15}}, \boldsymbol{v_{18}}) < T$, and $T$ is updated to $\delta(\boldsymbol{v_{0}}, \boldsymbol{v_{18}})=0.49$. This process continues, resulting in the replacement of $\boldsymbol{v_0}$ and $\boldsymbol{v_{15}}$ with $\boldsymbol{v_{17}}$ and $\boldsymbol{v_{3}}$, respectively. At the end of CA, $\boldsymbol{C(v_{18})}$ comprises $\{\boldsymbol{v_{17}}, \boldsymbol{v_5}, \boldsymbol{v_3}, \boldsymbol{v_{9}} \}$. The vertices $\boldsymbol{v_0}$, $\boldsymbol{v_3}$, $\boldsymbol{v_5}$, $\boldsymbol{v_9}$, $\boldsymbol{v_{13}}$, $\boldsymbol{v_{15}}$, and $\boldsymbol{v_{17}}$ are traversed in this phase. As shown in the graph index's memory layout (lower left), these vertices (highlighted in red) exhibit a random distribution, meaning numerous random memory accesses.
    In the NS step, $\boldsymbol{v_{17}}$ is initially chosen, leading to the removal of $\boldsymbol{v_3}$ and $\boldsymbol{v_9}$ from $\boldsymbol{C(v_{18})}$ based on the conditions $\delta(\boldsymbol{v_{17}},\boldsymbol{v_3}) < \delta(\boldsymbol{v_{3}},\boldsymbol{v_{18}})$ and $\delta(\boldsymbol{v_{17}},\boldsymbol{v_9}) < \delta(\boldsymbol{v_{9}},\boldsymbol{v_{18}})$. The remaining vertex, $\boldsymbol{v_5}$, is added to the neighbor list, which reaches the limit $R=2$, yielding $\{ \boldsymbol{v_{17}}, \boldsymbol{v_5}\}$. $\boldsymbol{v_{18}}$ is then appended to the neighbor lists of $\boldsymbol{v_{17}}$ and $\boldsymbol{v_5}$, while ensuring that their neighbor counts remain below $R$ through the heuristic strategy. In practice, $\boldsymbol{C(v_{18})}$ can include thousands of candidates, making full vector caching impractical and necessitating many random memory accesses for distance computations during the NS stage. To leverage SIMD instructions, a vector is divided into multiple segments, each comprising four dimensions suitable for loading into a 128-bit register. While SIMD enhances processing speed, the necessity of numerous register load operations to sequentially fetch each vector segment leads to reduced SIMD utilization efficiency.
\end{myExa}

Our analysis reveals that the primary inefficiency in HNSW index construction lies in the current distance computation process, consuming over 90\% of indexing time (Figure \ref{fig: indexing profile}). This process suffers from numerous random memory accesses and suboptimal SIMD utilization, misaligning with modern CPU architecture. Thus, optimizing distance computation to better exploit modern CPU capabilities is essential for accelerating HNSW construction.
\section{Methodology}
\label{sec: data processing}
In this section, we theoretically analyze how distance comparisons in \method{HNSW} construction can be effectively executed using compact vector codes. Based on this, we integrate three common vector compression methods—Product Quantization (PQ) \cite{PQ}, Scalar Quantization (SQ) \cite{LVQ}, and Principal Component Analysis (PCA) \cite{PCA}—into \method{HNSW}. While many variants of these methods exist \cite{OPQ, RaBitQ, WangLKC16, Distill-VQ, ZhangDW14, HeNB21, ScaNN, DeltaPQ}, we exclude them due to their complexity compared to the original methods, which complicates deployment in HNSW. We focus on these three methods for their lightweight nature and alignment with our goal of accelerating index construction. Drawing from the integration of these methods in HNSW construction, we design a novel compact coding strategy and optimize the memory layout for HNSW construction, enhancing memory access efficiency and SIMD operations.

\subsection{Theoretical Analysis}
Recall that distance computation during HNSW construction is a major bottleneck, with one distance value primarily compared against another. We denote distance comparisons for updating the candidate set as DC1, and for selecting final neighbors as DC2. Here, we demonstrate that DC1 and DC2 can be unified into a generalized framework and analyze how these comparisons can be effectively performed using compact vector codes.

\begin{myLemma}
\label{lemma: dist comp}
Given any three vertices $\boldsymbol{u}$, $\boldsymbol{v}$, and $\boldsymbol{w}$ in $D$-dimensional Euclidean space $\mathbb{R}^D$, the comparison between $\delta(\boldsymbol{u}, \boldsymbol{v})$ and $\delta(\boldsymbol{u}, \boldsymbol{w})$ is:

\noindent $\bullet$ $\delta(\boldsymbol{u}, \boldsymbol{v}) < \delta(\boldsymbol{u}, \boldsymbol{w})$ if and only if $\boldsymbol{e} \cdot \boldsymbol{u} - b < 0$;

\noindent $\bullet$ $\delta(\boldsymbol{u}, \boldsymbol{v}) > \delta(\boldsymbol{u}, \boldsymbol{w})$ if and only if $\boldsymbol{e} \cdot \boldsymbol{u} - b > 0$;

\noindent $\bullet$ $\delta(\boldsymbol{u}, \boldsymbol{v}) = \delta(\boldsymbol{u}, \boldsymbol{w})$ if and only if $\boldsymbol{e} \cdot \boldsymbol{u} - b = 0$;

\noindent where $\boldsymbol{e} \cdot \boldsymbol{u} - b = 0$ defines the perpendicular bisector hyperplane between $\boldsymbol{v}$ and $\boldsymbol{w}$, with $\boldsymbol{e}=\boldsymbol{w} - \boldsymbol{v}$ and $b = \frac{\|\boldsymbol{w}\|^2 - \|\boldsymbol{v}\|^2}{2}$.
\end{myLemma}

\begin{proof}
To show $\delta(\boldsymbol{u}, \boldsymbol{v}) = \delta(\boldsymbol{u}, \boldsymbol{w})$, we square both sides to eliminate square roots, giving $\|\boldsymbol{u} - \boldsymbol{v} \|^2 = \| \boldsymbol{u} - \boldsymbol{w} \|^2$. Expanding and simplifying, we get $2\boldsymbol{u} \cdot (\boldsymbol{w} - \boldsymbol{v}) = \|\boldsymbol{w}\|^2 - \|\boldsymbol{v}\|^2$.
Defining $\boldsymbol{e} = \boldsymbol{w} - \boldsymbol{v}$ and $b = \frac{\|\boldsymbol{w}\|^2 - \|\boldsymbol{v}\|^2}{2}$, this simplifies to $\boldsymbol{e} \cdot \boldsymbol{u} = b$, the equation of a hyperplane in $\mathbb{R}^D$. For $\boldsymbol{u}$ to lie on this hyperplane, it must satisfy $\boldsymbol{e} \cdot \boldsymbol{u} - b = 0$. Similarly, for \(\delta(\boldsymbol{u}, \boldsymbol{v}) > \delta(\boldsymbol{u}, \boldsymbol{w})\), squaring both sides gives $\|\boldsymbol{u} - \boldsymbol{v}\|^2 > \|\boldsymbol{u} - \boldsymbol{w}\|^2$, leading to $2\boldsymbol{u} \cdot (\boldsymbol{w} - \boldsymbol{v}) > \|\boldsymbol{w}\|^2 - \|\boldsymbol{v}\|^2$, and thus $\boldsymbol{e} \cdot \boldsymbol{u} - b > 0$. For \(\delta(\boldsymbol{u}, \boldsymbol{v}) < \delta(\boldsymbol{u}, \boldsymbol{w})\), a similar process gives $\boldsymbol{e} \cdot \boldsymbol{u} - b < 0$.
\end{proof}

In DC1, with $\boldsymbol{u}$ as the newly inserted vector and $\boldsymbol{v}$ as the farthest vertex from $\boldsymbol{u}$ in the candidate set $\boldsymbol{C(u)}$, we update $\boldsymbol{C(u)}$ by including $\boldsymbol{w}$ if $\boldsymbol{e} \cdot \boldsymbol{u} - b > 0$. In DC2, with $\boldsymbol{v}$ as the newly inserted vector and $\boldsymbol{w}$ as a selected neighbor in $\boldsymbol{N(v)}$, we exclude a candidate $\boldsymbol{u}$ from consideration as a neighbor if $\boldsymbol{e} \cdot \boldsymbol{u} - b > 0$. Using $\boldsymbol{e} \cdot \boldsymbol{u} - b > 0$ as a criterion, both DC1 and DC2 can be correctly executed.

\begin{myTheorem}[\cite{LVQ}]
\label{the: Dist. Comp.}
{
For three vertices $\boldsymbol{u}$, $\boldsymbol{v}$, and $\boldsymbol{w}$ in $D$-dimensional Euclidean space $\mathbb{R}^D$, with compact vector codes $\boldsymbol{u}^{\prime}$, $\boldsymbol{v}^{\prime}$, and $\boldsymbol{w}^{\prime}$, the condition $\delta(\boldsymbol{u}, \boldsymbol{v}) > \delta(\boldsymbol{u}, \boldsymbol{w})$ is equivalent to $\delta(\boldsymbol{u}^{\prime}, \boldsymbol{v}^{\prime}) > \delta(\boldsymbol{u}^{\prime}, \boldsymbol{w}^{\prime})$ when $|\boldsymbol{e} \cdot \boldsymbol{u} - b|\geq |E|$. The quantity $E$ is defined as
\begin{equation}
\begin{aligned}
    E = &(E_{\boldsymbol{w}}-E_{\boldsymbol{v}}) \cdot \boldsymbol{u} + (\boldsymbol{w}-\boldsymbol{v})\cdot E_{\boldsymbol{u}} + E_{\boldsymbol{v}} \cdot E_{\boldsymbol{u}} - E_{\boldsymbol{w}}\cdot E_{\boldsymbol{u}}\\
    &+ \frac{1}{2}||E_{\boldsymbol{w}}||^2-\frac{1}{2}||E_{\boldsymbol{v}}||^2 +\boldsymbol{v}\cdot E_{\boldsymbol{v}}-\boldsymbol{w}\cdot E_{\boldsymbol{w}}
\end{aligned}
\end{equation}
where $E_{\boldsymbol{u}}$, $E_{\boldsymbol{v}}$, and $E_{\boldsymbol{w}}$ are the error vectors corresponding to $\boldsymbol{u}$, $\boldsymbol{v}$, and $\boldsymbol{w}$, respectively.
}

\end{myTheorem}

\begin{proof}
    Based on Lemma \ref{lemma: dist comp}, the condition $\delta(\boldsymbol{u}^{\prime}, \boldsymbol{v}^{\prime}) > \delta(\boldsymbol{u}^{\prime}, \boldsymbol{w}^{\prime})$ is equivalent to $\boldsymbol{e}^{\prime} \cdot \boldsymbol{u}^{\prime} - b^{\prime} > 0$, where $\boldsymbol{e}^{\prime}=\boldsymbol{w}^{\prime}-\boldsymbol{v}^{\prime}$ and $b^{\prime}=\frac{\|\boldsymbol{w}^{\prime}\|^2 - \|\boldsymbol{v}^{\prime}\|^2}{2}$. To demonstrate that $\delta(\boldsymbol{u}, \boldsymbol{v}) > \delta(\boldsymbol{u}, \boldsymbol{w})$ is equivalent to $\delta(\boldsymbol{u}^{\prime}, \boldsymbol{v}^{\prime}) > \delta(\boldsymbol{u}^{\prime}, \boldsymbol{w}^{\prime})$, it suffices to prove that $\boldsymbol{e} \cdot \boldsymbol{u} - b$ and $\boldsymbol{e}^{\prime} \cdot \boldsymbol{u}^{\prime} - b^{\prime}$ share the same sign.
Substituting $\boldsymbol{e}^{\prime}$ and $b^{\prime}$ into $\boldsymbol{e}^{\prime} \cdot \boldsymbol{u}^{\prime} - b^{\prime}$, we obtain:
\begin{equation}
\begin{aligned}
    \boldsymbol{e}^{\prime} \cdot \boldsymbol{u}^{\prime} - b^{\prime} = (\boldsymbol{w}^{\prime}-\boldsymbol{v}^{\prime}) \cdot \boldsymbol{u}^{\prime} - \frac{\|\boldsymbol{w}^{\prime}\|^2 - \|\boldsymbol{v}^{\prime}\|^2}{2}
\end{aligned}
\end{equation}
For vector $\boldsymbol{u}$, let $E_{\boldsymbol{u}}$ denote the compression error of $\boldsymbol{u}$, defined as
\begin{equation}
\label{equ: error}
    E_{\boldsymbol{u}} = \boldsymbol{u} - \boldsymbol{u}^{\prime}
\end{equation}
Substituting this into the previous equation, we derive:
\begin{equation}
\label{equ: error1}
\begin{aligned}
    \boldsymbol{e}^{\prime} \cdot \boldsymbol{u}^{\prime} - b^{\prime}
    = &((\boldsymbol{w}-E_{\boldsymbol{w}})-(\boldsymbol{v}-E_{\boldsymbol{v}})) \cdot (\boldsymbol{u}-E_{\boldsymbol{u}}) \\
    &- \frac{\|\boldsymbol{w}-E_{\boldsymbol{w}}\|^2 - \|\boldsymbol{v}-E_{\boldsymbol{v}}\|^2}{2}
\end{aligned}
\end{equation}
Expanding and simplifying, we group terms involving compression errors:
\begin{equation}
\begin{aligned}
    E = &(E_{\boldsymbol{w}}-E_{\boldsymbol{v}}) \cdot \boldsymbol{u} + (\boldsymbol{w}-\boldsymbol{v})\cdot E_{\boldsymbol{u}} + E_{\boldsymbol{v}} \cdot E_{\boldsymbol{u}} - E_{\boldsymbol{w}}\cdot E_{\boldsymbol{u}}\\
    &+ \frac{1}{2}||E_{\boldsymbol{w}}||^2-\frac{1}{2}||E_{\boldsymbol{v}}||^2 +\boldsymbol{v}\cdot E_{\boldsymbol{v}}-\boldsymbol{w}\cdot E_{\boldsymbol{w}}
\end{aligned}
\end{equation}
Thus, the equation simplifies to:
\begin{equation}
\label{equ: error2}
\begin{aligned}
    \boldsymbol{e}^{\prime} \cdot \boldsymbol{u}^{\prime} - b^{\prime}
    &= (\boldsymbol{w}-\boldsymbol{v}) \cdot \boldsymbol{u} -\frac{1}{2} (||\boldsymbol{w}||^2-||\boldsymbol{v}||^2) - E \\
    &=\boldsymbol{e}\cdot \boldsymbol{u}-b - E
\end{aligned}
\end{equation}
Therefore, when $|\boldsymbol{e} \cdot \boldsymbol{u} - b|\geq |E|$, the sign of $\boldsymbol{e}^{\prime} \cdot \boldsymbol{u}^{\prime} - b^{\prime}$ matches that of $\boldsymbol{e} \cdot \boldsymbol{u}-b$. Specifically, $\boldsymbol{e}^{\prime} \cdot \boldsymbol{u}^{\prime} - b^{\prime}>0$ if and only if $\boldsymbol{e} \cdot \boldsymbol{u}-b>0$, and similarly, $\boldsymbol{e}^{\prime} \cdot \boldsymbol{u}^{\prime} - b^{\prime}<0$ if and only if $\boldsymbol{e} \cdot \boldsymbol{u}-b<0$.
\end{proof}

Theorem \ref{the: Dist. Comp.} elucidates how distance comparison is preserved despite compression error. By adjusting the parameters of compact coding (thus affecting $E$), the condition $|\boldsymbol{e} \cdot \boldsymbol{u} - b|\geq |E|$ can always be satisfied. Consequently, an appropriate compression error maintains accurate distance comparison while reducing vector data size, enhancing memory access and SIMD operations. Hence, integrating vector compression techniques into the HNSW construction process is a promising avenue for accelerating index construction.
{
Given a vector $\boldsymbol{u}$ and its compact vector code $\boldsymbol{u}^{\prime}$, the error vector is computed as $E_{\boldsymbol{u}} = \boldsymbol{u} - \boldsymbol{u}^{\prime}$. Here, $\boldsymbol{u}^{\prime}$ refers to the vector derived from the compact vector code, not the vector code itself; this derived vector retains the same dimensionality as $\boldsymbol{u}$.
For different compact coding methods, the derived vectors have different interpretations, and thus $E_{\boldsymbol{u}}$ is computed in distinct ways (see Section \ref{subsec: baseline solutions}).
}

{
In a compact coding method, adjustable parameters control the compression error. Specifically, a random sample of vectors (e.g., 10,000) is drawn from the dataset, and each vector's top-100 nearest neighbors are determined, forming a triple $(\boldsymbol{u}$, $\boldsymbol{v}$, $\boldsymbol{w})$ from the vector and its two nearest neighbors. A batch of such triples is then constructed, followed by the generation of compact vector codes $\boldsymbol{u}^{\prime}$, $\boldsymbol{v}^{\prime}$, and $\boldsymbol{w}^{\prime}$ for specific parameters. Next, the quantities $|\boldsymbol{e} \cdot \boldsymbol{u} - b|$ and $|E|$ are computed for each triple. By tuning the parameters, one can maximize the proportion of triples that satisfy $|\boldsymbol{e} \cdot \boldsymbol{u} - b| \geq |E|$ while minimizing the vector size. Finally, the chosen parameters are evaluated by integrating the compact coding method into HNSW.
}

\subsection{Baseline Solutions}
\label{subsec: baseline solutions}

\subsubsection{\textbf{Deploying PQ in HNSW}}
\label{subsubsec: HNSW PQ}
Product Quantization (PQ) \cite{PQ} splits high-dimensional vectors into subvectors that are independently compressed. For each subspace, a codebook of centroids is generated, and during encoding, the ID of the nearest centroid is selected to create a compact representation. The trade-off between distance comparison accuracy and compression error is managed by adjusting the code length ($L_{PQ}$) in each subspace and the number of subspaces ($M_{PQ}$). {For example, with $M_{PQ} = 8$ and $L_{PQ} = 4$, a vector is represented by 8 codewords, each storing a 4-bit centroid ID in the subspace. For a vector $\boldsymbol{u}$ and its compact code $\boldsymbol{u}^{\prime}$, the derived vector is formed by concatenating $\boldsymbol{u}$'s nearest centroid vectors (identified by $\boldsymbol{u}^{\prime}$) from all subspaces. The error vector $E_{\boldsymbol{u}}$ is then computed as the difference between $\boldsymbol{u}$ and this derived vector.}

To incorporate PQ into HNSW, we preprocess high-dimensional vectors to create codebooks. Each inserted vector is encoded using these codebooks, and a distance table is computed between the vector and subspace centroids (Asymmetric Distance Computation, ADC \cite{PQ}). In the Candidate Acquisition (CA) stage, distances to visited vertices are efficiently computed by scanning this table. In the Neighbor Selection (NS) stage, the table cannot compute candidate distances. Thus, centroids are located based on compact PQ codes, and inter-centroid distances represent candidate distances (Symmetric Distance Computation, SDC \cite{PQ}). Theorem \ref{the: Dist. Comp.} suggests that optimal values of $L_{PQ}$ and $M_{PQ}$ balance distance comparison accuracy with construction efficiency. We denote the HNSW constructed via this pipeline as HNSW-PQ, which replaces original vectors with PQ codes and employs ADC and SDC for efficient distance computation during index construction.

We evaluate HNSW-PQ under different $L_{PQ}$ and $M_{PQ}$ configurations, as 
shown in Figure \ref{fig: HNSW-PQ laion1m}. The results indicate that both parameters significantly influence indexing time and search performance. As $L_{PQ}$ grows, more clusters are formed in each subspace, leading to longer codebook generation and vector encoding times, thus increasing indexing time. Notably, indexing time initially decreases before rising with $M_{PQ}$. This occurs because a small $M_{PQ}$ causes high subspace dimensionality, increasing codebook generation time. Conversely, a larger $M_{PQ}$ also raises indexing time due to more subspaces. Regarding index quality, both higher $L_{PQ}$ and $M_{PQ}$ improve recall, aligning with the theoretical analysis that lower compression errors yield more accurate distance comparisons.

\begin{figure}
\setlength{\abovecaptionskip}{0cm}
  \setlength{\belowcaptionskip}{0cm}
  \centering
  \footnotesize
  \stackunder[0.5pt]{\includegraphics[scale=0.2]{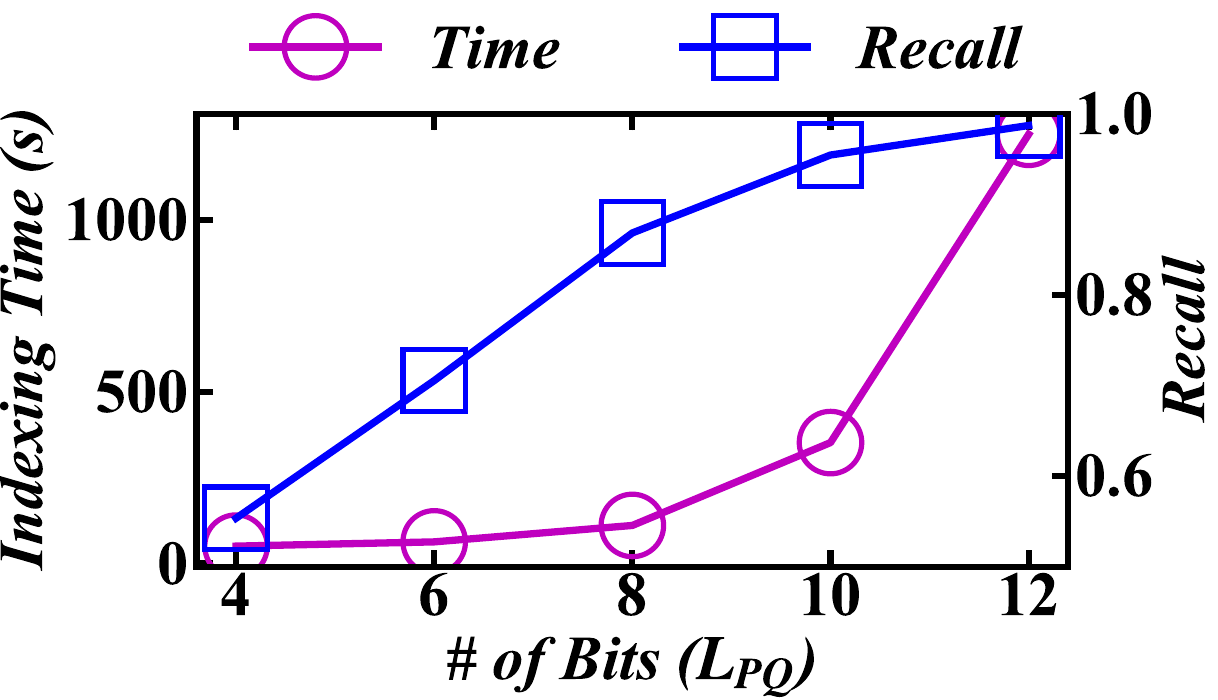}}{(a) LAION-1M ($M_{PQ}=8$)}
  \hspace{0.15cm}
  \stackunder[0.5pt]{\includegraphics[scale=0.2]{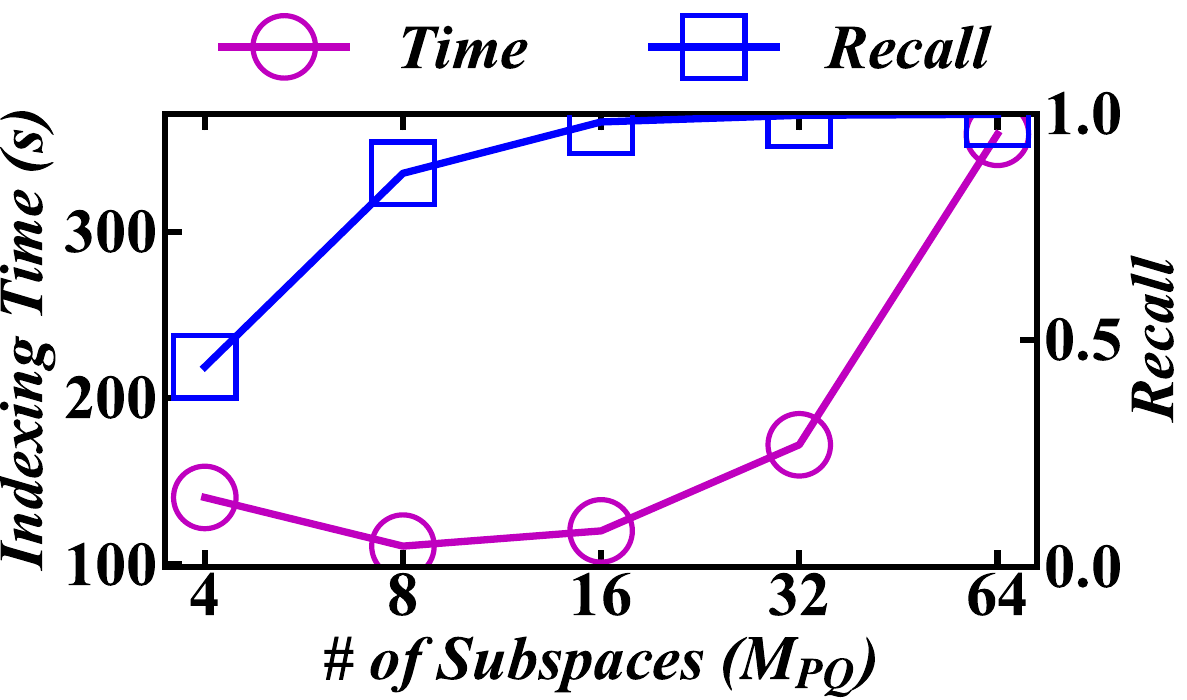}}{(b) LAION-1M ($L_{PQ}=8$)}
  \newline
  \caption{Effect of parameters on HNSW-PQ.}
  \label{fig: HNSW-PQ laion1m}
\end{figure}

\subsubsection{\textbf{Deploying SQ in HNSW}}
\label{subsubsec: HNSW SQ}
Scalar Quantization (SQ) \cite{LVQ,Boufounos12} compresses each dimension of vectors independently. It divides data ranges into intervals, assigning each interval an integer value. Floating-point values within an interval are approximated by that integer value, thus reducing vector size by lowering the required bits ($L_{SQ}$). {For example, with $L_{SQ} = 8$, each dimension is represented by an integer (0$\sim$255). Before distance computation, these integers are converted back into floating-point numbers to form the decoded vector, which is lossy. The error vector $E_{\boldsymbol{u}}$ is obtained by subtracting the decoded vector from $\boldsymbol{u}$.}

To accelerate HNSW construction via SQ, we preprocess high-dimensional vectors to identify the maximum and minimum values per dimension. We calculate the interval for each dimension by subtracting the minimum from the maximum. During index construction, each inserted vector is encoded using these precomputed intervals, producing compact SQ codes for distance computations. As per Theorem \ref{the: Dist. Comp.}, the optimal $L_{SQ}$ balances distance comparison accuracy with construction efficiency. The resulting HNSW, termed HNSW-SQ, replaces original vector data with compact SQ codes, enabling efficient distance calculations.

In Figure \ref{fig: HNSW-SQ and HNSW-PCA laion1m}(a), we assess HNSW-SQ across varying $L_{SQ}$ values. The results show that indexing time first decreases and then increases as $L_{SQ}$ grows. This trend arises from the lack of specialized data types for 2 and 4 bits, necessitating the use of \textit{uint8} type, which reduces operational efficiency. In contrast, 8 bits, while larger, aligns well with \textit{uint8}, achieving an optimal balance between compression error and operational efficiency, thereby reducing indexing time. While 16 bits can be represented with \textit{uint16}, the larger vector size increases indexing time. Notably, search accuracy consistently improves with higher $L_{SQ}$, supporting the theoretical analysis.

\subsubsection{\textbf{Deploying PCA in HNSW}}
\label{subsubsec: HNSW PCA}
Principal Component Analysis (PCA) \cite{PCA} reduces dimensionality by projecting high-dimensional vectors into a lower-dimensional space. {It applies an orthogonal transformation matrix to a vector $\boldsymbol{u}$, preserving its norm while prioritizing significant components in lower dimensions. The compact vector code $\boldsymbol{u}^{\prime}$ is formed by retaining the first $d_{PCA}$ dimensions, substantially reducing storage requirements. The error vector $E_{\boldsymbol{u}}$ is calculated by subtracting the dimension-reduced vector from the transformed vector of $\boldsymbol{u}$, with the dimension-reduced vector zero-padded to match dimensionality.}

To expedite HNSW construction via PCA, we preprocess high-dimensional vectors by deriving an orthogonal projection matrix through eigenvalue decomposition of the covariance matrix. Each inserted vector is then transformed into a low-dimensional space using this matrix, with principal components serving as compact PCA codes. Distance computations are performed directly on these compact representations. Theorem \ref{the: Dist. Comp.} indicates that an optimal $d_{PCA}$ balances distance comparison accuracy and construction efficiency. The resulting HNSW, designated as HNSW-PCA, substitutes the original vectors with compact PCA representations, thereby facilitating more efficient distance computations.

Figure \ref{fig: HNSW-SQ and HNSW-PCA laion1m}(b) shows the evaluation of HNSW-PCA across various $d_{PCA}$ configurations. We find that indexing time increases with higher $d_{PCA}$, while search accuracy improves, aligning with the theoretical analysis. Optimal trade-off between indexing time and search performance occurs with $d_{PCA}$ ranging from 256 to 512. Notably, $d_{PCA}$ reaches 420 when 90\% cumulative variance is achieved on the LAION dataset. In the experiments, $d_{PCA}$ will be set to ensure at least 90\% cumulative variance.

\begin{figure}
  \setlength{\abovecaptionskip}{0cm}
  \setlength{\belowcaptionskip}{0cm}
  \centering
  \footnotesize
  \stackunder[0.5pt]{\includegraphics[scale=0.2]{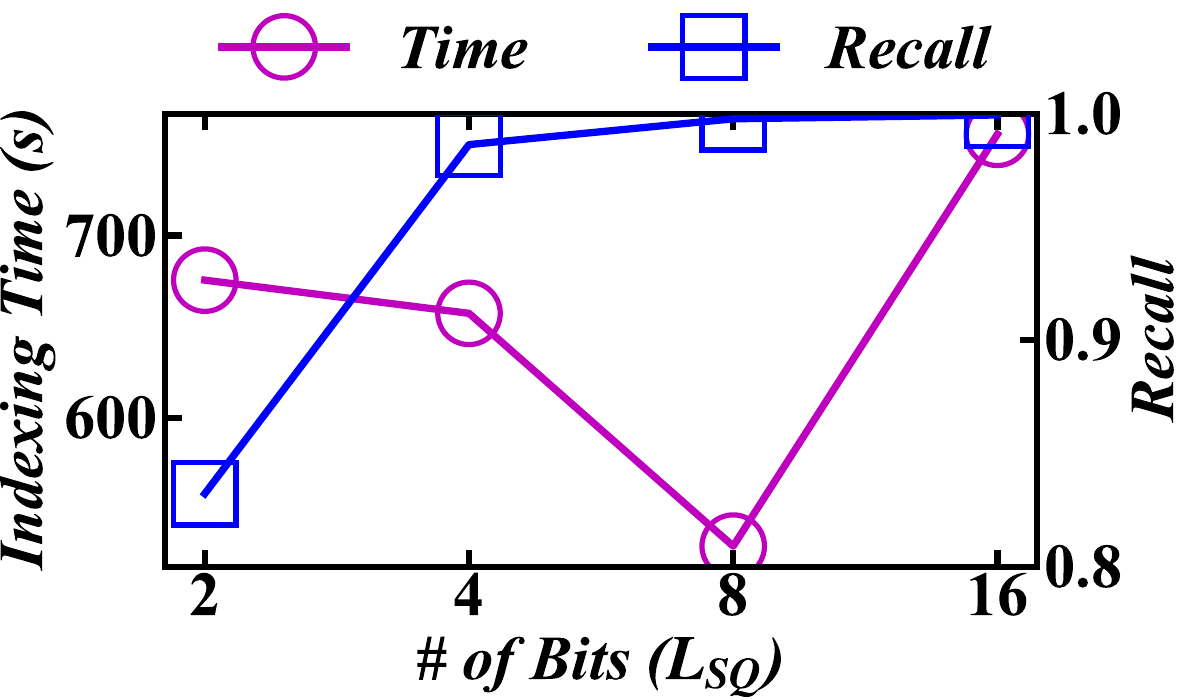}}{(a) LAION-1M (HNSW-SQ)}
  \hspace{0.15cm}
  \stackunder[0.5pt]{\includegraphics[scale=0.2]{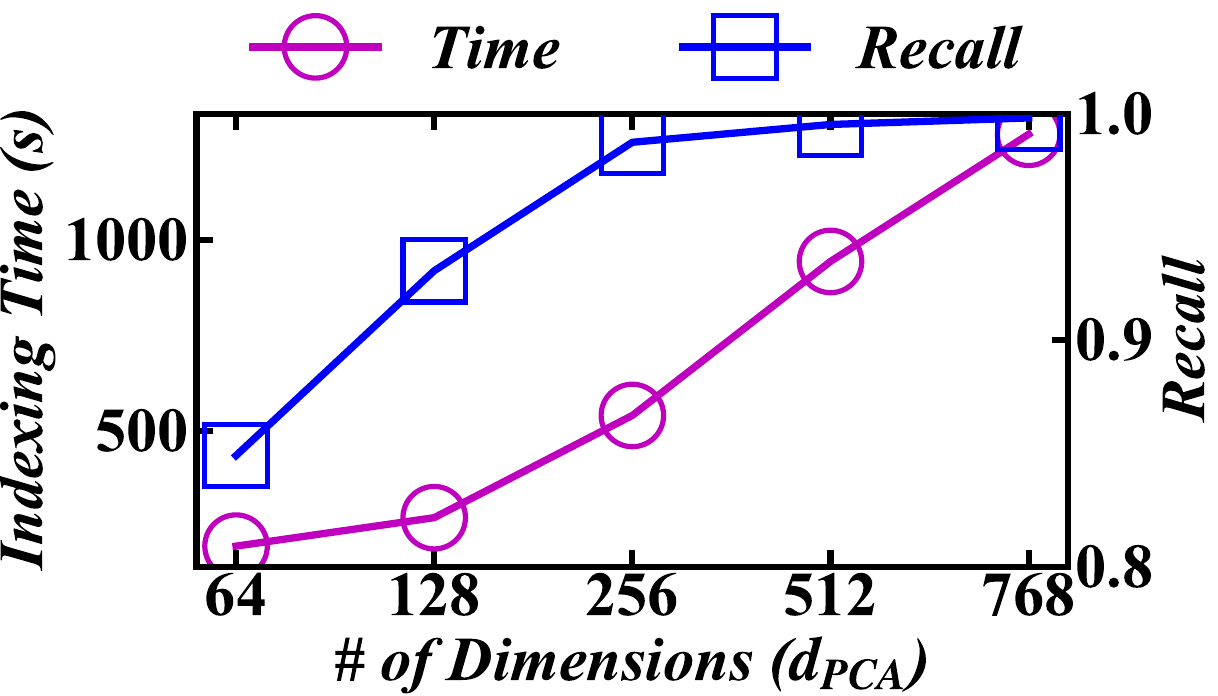}}{(b) LAION-1M (HNSW-PCA)}
  \newline
  \caption{Effect of parameters on HNSW-SQ and HNSW-PCA.}
  \label{fig: HNSW-SQ and HNSW-PCA laion1m}
\end{figure}

\vspace{-0.15cm}
\subsubsection{\textbf{Lessons Learned}}
\label{subsubsec: lessons learned}
Theoretical and empirical analyses reveal that index construction can be accelerated using compact coding techniques, and construction efficiency and index quality can be well balanced by adjusting the compression error. An appropriate compression error reduces vector data size, improving memory access and SIMD operation efficiency while preserving accurate distance comparisons. Following this insight, optimized variants \cite{OPQ, RaBitQ, WangLKC16, Distill-VQ, ZhangDW14, HeNB21, ScaNN, DeltaPQ, Boufounos12} of PQ, SQ, and PCA may be integrated into HNSW to further speed up index construction. Note that these variants must avoid excessive processing overhead to maintain the balance between construction efficiency and index quality. Nonetheless, the core mechanism of these methods—vector size reduction—does not align well with HNSW's construction characteristics on modern CPUs. Regardless of whether HNSW-PQ, HNSW-SQ, or HNSW-PCA is used, the random memory access pattern during index construction remains unchanged. In terms of arithmetic operations, HNSW-PQ’s distance table lookup is unable to fully utilize SIMD operations, as each computation is executed individually \cite{RaBitQ}. While HNSW-SQ and HNSW-PCA reduce register loads by computing directly on compressed vectors, they still process distance computations sequentially, underutilizing SIMD’s advantages \cite{PQfast}. We highlight that current variants of PQ, SQ, and PCA do not fundamentally resolve these limitations.

\subsection{The \texttt{Flash} Method}
\label{subsec: opt method}
In Section \ref{subsec: baseline solutions}, we demonstrate that incorporating existing compact coding methods into HNSW construction can enhance index construction speed. However, these methods either offer limited gains in indexing efficiency or degrade search performance. This is mainly because they are unable to effectively reduce random memory accesses and fully exploit the advantages of SIMD instructions. To overcome these limitations, we propose a speci\underline{\textbf{f}}ic compact coding strategy a\underline{\textbf{l}}ongside memory l\underline{\textbf{a}}yout optimization for the con\underline{\textbf{s}}truction process of \underline{\textbf{H}}NSW, named \texttt{Flash}, which minimizes random memory accesses while maximizing SIMD utilization.

\begin{figure*}
  \centering
  \setlength{\abovecaptionskip}{0cm}
  \setlength{\belowcaptionskip}{-0.1cm}
  \includegraphics[width=\linewidth]{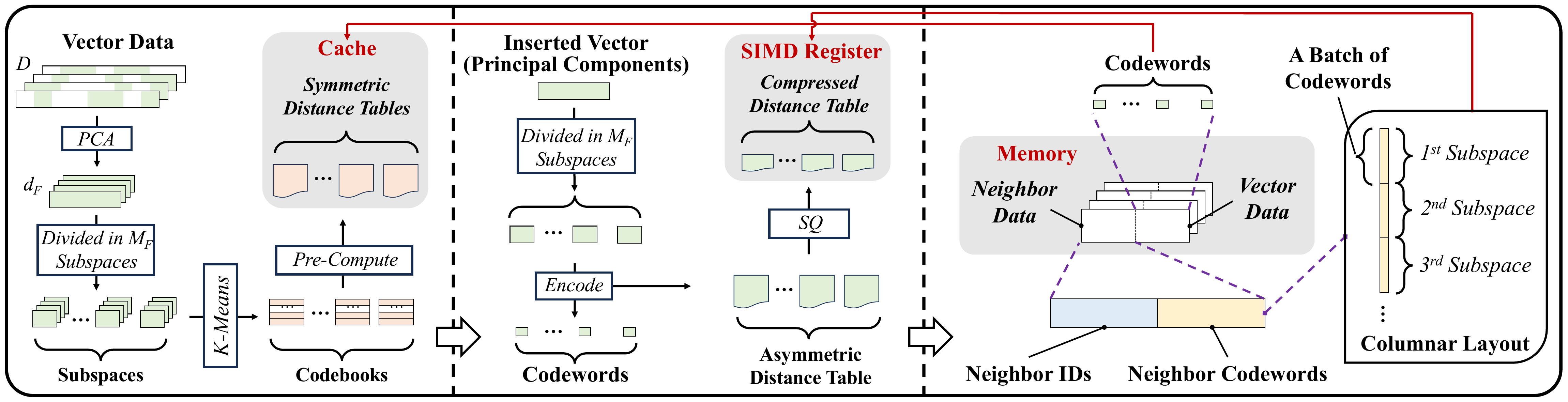}
  \caption{Illustrating the coding pipeline and data layout of \texttt{Flash}.}
  \label{fig:compact code}
  \vspace{-0.3cm}
\end{figure*}

\subsubsection{\textbf{Design Overview}}
\label{subsubsec: design overview}
Noting that HNSW-PQ strikes an optimal balance between construction efficiency and search performance, with PQ offering flexibility through two adjustable parameters that affect the compression error. We leverage the core principles of PQ to devise \texttt{Flash}, tailored to the HNSW construction process on modern CPUs. Figure \ref{fig:compact code} provides an overview of the coding pipeline and data layout used by \texttt{Flash}.

Given that a CPU core has multiple SIMD registers, \texttt{Flash} partitions the high-dimensional space into subspaces, enabling parallel processing of subspaces in these registers.
To accommodate the distinct distance computations in Candidate Acquisition (CA) and Neighbor Selection (NS) stages, it introduces Asymmetric Distance Tables (ADT) for the CA stage, residing in SIMD registers, and Symmetric Distance Tables (SDT) for the NS stage, located in cache. It optimizes the bit counts allocated for each codeword (i.e., encoded vector) and applies SQ to compress the precomputed ADT to fit within a SIMD register, thereby reducing register loads. Since high-dimensional vectors often exhibit uneven variances across dimensions, directly encoding original subspaces may result in suboptimal bit utilization. To address this, \texttt{Flash} utilizes PCA to extract the principal components of vectors, generating subspaces based on these components, enhancing bit utilization within the SIMD register's constraints.
When organizing neighbor lists, neighbor IDs are grouped with corresponding codewords to minimize random memory accesses. Rather than sequentially storing all codewords for a neighbor, \texttt{Flash} gathers the codewords of a batch of neighbors within a specific subspace. The batch size matches the number of centroids per subspace. This alignment allows a register load to fetch an entire batch, enabling partial distances to be obtained in parallel using SIMD shuffle operations since the ADT for a subspace is fully contained within a SIMD register.

\texttt{Flash} enhances HNSW construction by optimizing memory layout and SIMD operations. Two hyperparameters—the number of subspaces ($M_F$) and the dimension of the principal components ($d_F$)—can be adjusted to balance construction efficiency and search performance (Theorem \ref{the: Dist. Comp.}). These adjustments do not affect \texttt{Flash}'s memory access pattern or SIMD optimizations. We highlight that existing SIMD optimizations for PQ, as discussed in \cite{PQfast,QuickADC}, target linear scans or inverted file (IVF) structures, where vector batches are naturally grouped, rather than graph index settings \cite{RaBitQ}.

\subsubsection{\textbf{Extracting Principal Components}}
\label{subsubsec: extract principal components}
To align with SIMD register constraints, each subspace has a limited bit representation (e.g., 4 bits). Encoding on principal components optimally utilizes these bits, reducing quantization error.
Thus, \texttt{Flash} employs PCA to project high-dimensional vectors by rearranging dimensions in descending order of variance. This approach preserves essential low-dimensional principal components by selecting the first $d_{F}$ dimensions, encapsulating the most significant features.
Given a dataset $\boldsymbol{S}$ of $n$ vectors, each with $D$ dimensions, the mean vector is $\bar{\boldsymbol{u}} = \frac{1}{n} \sum_{i=1}^{n} \boldsymbol{u}_i$. Data is centered by subtracting $\bar{\boldsymbol{u}}$ from each vector, producing the centered matrix $\dot{\boldsymbol{S}} = \boldsymbol{S} - \boldsymbol{1} \cdot \bar{\boldsymbol{u}}^{\top}$, where $\boldsymbol{1}$ is a unit vector.
The covariance matrix $\Sigma = \frac{1}{n} (\dot{\boldsymbol{S}}^{\top} \dot{\boldsymbol{S}})$ is then computed, and its eigenvalues $\lambda_1, \lambda_2, \ldots, \lambda_{D}$ and eigenvectors $\boldsymbol{a}_1, \boldsymbol{a}_2, \ldots, \boldsymbol{a}_{D}$ are extracted. For a given variance fraction $\alpha$, the function $f(d) = \frac{\sum_{i=1}^{d} \lambda_i}{\sum_{i=1}^{D} \lambda_i}$ identifies the smallest $d_{F}$ such that $f(d_{F}) \geq \alpha$, defining the principal component space of dimension $d_{F}$. The principal components of each vector are established based on the basis $\boldsymbol{A}_{1:d_{F}} = (\boldsymbol{a}_1 \ \boldsymbol{a}_2 \ \cdots \ \boldsymbol{a}_{d_{F}})$, yielding the reduced-dimensional data $\widetilde{\boldsymbol{S}}$:
\begin{equation}
 \widetilde{\boldsymbol{S}} = \{ \tilde{\boldsymbol{u}}_i \ | \ \tilde{\boldsymbol{u}}_i = \mathbf{A}_{1:d_{F}}^{\top} \boldsymbol{u}_i, \ \text{for } i = 1, \ldots, n \}
\end{equation}

\subsubsection{\textbf{Subspace Division and Distance Table Compression}}
\label{subsubsec: subspace division}
\texttt{Flash} decomposes each vector's principal components $\tilde{\boldsymbol{u}}$ into $M_{F}$ disjoint subvectors, denoted as $\tilde{\boldsymbol{u}} = [\tilde{\boldsymbol{u}}_{1}, \tilde{\boldsymbol{u}}_{2}, \ldots, \tilde{\boldsymbol{u}}_{M_{F}}]$, where each $\tilde{\boldsymbol{u}}_{i}$ resides in a specific subspace of the principal component domain. Similar to PQ, a codebook $\boldsymbol{C}_i = \{\boldsymbol{c}_{i1}, \boldsymbol{c}_{i2}, \ldots, \boldsymbol{c}_{iK}\}$ is created for each subspace, with $\boldsymbol{c}_{ij}$ denoting a codeword (i.e., a centroid ID) in the $i$-th subspace, and $K$ representing the number of centroids. The $i$-th subvector $\tilde{\boldsymbol{u}}_{i}$ is quantized by identifying the closest centroid in the codebook $\boldsymbol{C}_i$ of the $i$-th subspace:
\begin{equation}
\psi(\tilde{\boldsymbol{u}}_{i}) = \arg \min_{\boldsymbol{c}_{ij} \in \boldsymbol{C}_i} \delta( \tilde{\boldsymbol{u}}_{i}, \boldsymbol{c}_{ij} )
\end{equation}
where $\psi(\tilde{\boldsymbol{u}}_{i})$ is the nearest centroid's ID serving as the codeword of the subvector. Consequently, the complete principal components $\tilde{\boldsymbol{u}}$ are represented as a sequence of codewords, one for each subvector. The codeword length $L_F$ relates to the number of centroids $K$ via $L_F=\lceil log_2{K} \rceil$ in a subspace. \texttt{Flash} samples a subset from the complete vector set to efficiently construct codebooks, following PQ and its variants \cite{PQ,OPQ}.

For each inserted vector $\boldsymbol{u}$, asymmetric distance tables (ADT) and codewords are simultaneously generated across $M_{F}$ subspaces using precomputed codebooks, leveraging shared distance computations between each subvector $\tilde{\boldsymbol{u}}_{i}$ and the centroids in $\boldsymbol{C}_i$.
This approach avoids duplicate calculations, enhancing efficiency. In the Candidate Acquisition (CA) stage, the ADT enables fast summation of partial distances to calculate distances between $\boldsymbol{u}$ and visited vertices.
To adapt ADT for SIMD registers, \texttt{Flash} employs SQ to map each distance $dist$ in the table to discrete levels:
\begin{equation}
\eta(dist) = \Big\lfloor \left(\frac{dist - dist_{\min}}{\Delta}\right) \cdot (2^{H}-1) \Big\rfloor
\end{equation}
where $\eta(dist)$ is the quantized distance, $dist_{\min}$ is the minimum distance, $\Delta$ is the quantization step size ($\Delta=dist_{max}-dist_{min}$), and $H$ is the bit number for a quantized distance value. {With $K=16$ and $H=8$, each ADT is 128 bits in size. These ADTs are stored as \textit{one-dimensional arrays} that can reside in registers during the insertion process, and are freed once the insertion is completed.}

In the Neighbor Selection (NS) stage, distance computations among candidates preclude the use of ADT. To address this, \texttt{Flash} precomputes a symmetric distance table (SDT) that stores distances between centroids in each subspace. {With $M_F$ subspaces and $K$ centroids, $M_F$ SDTs are built, each holding $K^2$ distance values in a \textit{two-dimensional array}.} These values are quantized similarly to ADT for optimal caching. \texttt{Flash} obtains distances between candidates from SDTs based on their codewords. Note that SDTs are shared by all inserted vectors during index construction, eliminating numerous random memory accesses.
To compare distance values from SDTs and ADTs for neighbor selection, \texttt{Flash} uses the same quantization step size $\Delta$ and target bit number $H$ for both tables. It calculates the maximum and minimum distances ($dist_{max}^{i}$, $dist_{min}^{i}$) within each subspace, summing each $dist_{max}^{i}$ to obtain $dist_{max}$ and taking the lowest among $dist_{min}^{i}$ as $dist_{min}$.

\subsubsection{\textbf{Access-Aware Memory Layout}}
{In the graph index, all vertices share a uniform neighbor list structure and size. A \textit{contiguous memory block} is allocated for each vertex's vector and neighbor data, accessed via an offset determined by the vertex ID.}
For an inserted vector $\boldsymbol{u}$, the CA stage updates candidates via greedy search, visiting a vertex's neighbors in the sub-graph index and computing their distances to $\boldsymbol{u}$.
Current graph indexes suffer from numerous random accesses when fetching neighbor vectors, as neighbor vectors are too big to be stored alongside neighbor IDs. By attaching neighbor codewords to these IDs, \texttt{Flash} computes distances directly in register-resident ADTs, avoiding random memory accesses. {For each vertex, \texttt{Flash} stores the neighbor IDs in one fixed memory block, followed by the neighbor codewords in another (see Figure \ref{fig:compact code}, lower right).} To efficiently use SIMD instructions, \texttt{Flash} gathers $B$ neighbor codewords in a subspace, processing them concurrently to ensure efficient SIMD register loading and parallel distance computations for $B$ neighbors.

\subsubsection{\textbf{SIMD Acceleration}}
\label{subsubsec: simd acceler}
SIMD instructions support vectorized execution, enabling simultaneous processing of multiple data points. This can accelerate distance calculations between the inserted vector $\boldsymbol{u}$ and a batch of neighbors. \texttt{Flash} uses SIMD to look up ADTs and sum partial distances across subspaces. An ADT resides in a SIMD register, storing distances between $\boldsymbol{u}$ and centroids in each subspace. \texttt{Flash} employs shuffle operations to extract partial distances using neighbors’ codewords. In a single operation, codewords for $B$ neighbors are loaded into a register, which then serves as indices to access the partial distances in the ADT. Specifically, the shuffle operation rearranges the data within the register according to the indices, extracting partial distances from the ADT without complex conditional branching or excessive memory accesses. Partial distances from two subspaces are summed using SIMD instructions, and this process is repeated across all subspaces to compute the final distances for the $B$ neighbors. By using SIMD lookups for ADTs, \texttt{Flash} minimizes random memory accesses and exploits high-speed registers for parallel arithmetic, thereby improving computational efficiency.

\subsubsection{\textbf{Implementation on HNSW}}
\label{subsubsec: hnsw-flash}
We integrate \texttt{Flash} into the HNSW algorithm to accelerate indexing and search processes. It be- gins by preprocessing the dataset to extract principal components (Section \ref{subsubsec: extract principal components}), then sampling a subset, dividing it into subspaces, generating codebooks, and pre-computing centroid distances for the SDT. During index construction (lines 2-8 of Algorithm \hyperref[alg: hnsw construction]{1}), each vector $\boldsymbol{u}$ has its principal components partitioned by the codebooks. Distances between each subvector $\boldsymbol{u}_i$ and codebook centroids $\boldsymbol{C}_i$ form the ADT, and the nearest centroid ID is selected as the codeword for $\boldsymbol{u}_i$. The ADT is quantized (Section \ref{subsubsec: subspace division}) and held in a register until insertion completes. For line 5, SIMD acceleration computes distances between $\boldsymbol{u}$ and batches of neighbors based on ADT (Section \ref{subsubsec: simd acceler}). For lines 6–7, distances between candidates are approximated via SDT lookups, using each candidate’s codewords as indices. We tune the number of subspaces ($M_F$) and the dimensionality of principal components ($d_F$) to manage compression error (Theorem \ref{the: Dist. Comp.}) without affecting \texttt{Flash}’s efficient memory access or SIMD operations. Other fixed parameters include bits per codeword ($L_F$) and batch size ($B$), set by the SIMD register size. The search procedure follows the CA paradigm, and we apply an additional reranking step based on original vectors to refine results.

\vspace{0.2cm}
\noindent\underline{Remarks.}
(1) {Although many studies optimize the compression error of PQ, SQ, and PCA \cite{OPQ,ScaNN,LVQ,ZhangTHW22}, none are specifically tailored to meet the distinct requirements of graph indexing on modern CPU architectures. Integrating these compression methods into the graph construction process often yields limited efficiency gains or even degrades search performance, as they neither address random memory access patterns nor leverage SIMD efficiently. In contrast, \texttt{Flash} is a specialized compact coding strategy designed for graph index construction and optimized for current CPU architectures. Notably, \texttt{Flash} incorporates the complementary principles of established vector compression methods (e.g., PQ, SQ, PCA) within its framework and offers the flexibility to integrate new optimizations of these compression methods. While several studies integrate vector compression into search processes \cite{yue2023routing,DiskANN,Starling,LVQ}, search differs markedly from index construction, which involves more complex tasks such as the NS stage. Moreover, it is crucial that compression methods avoid excessive overhead, maintaining a balance between construction speed and index quality. Recent work \cite{yue2023routing,DiskANN,Starling,yang2024bridging} shows that boosting search performance often increases indexing time.}
(2) In \texttt{Flash}, codeword and ADT generations share identical distance computations (Section \ref{subsubsec: subspace division}), prompting an integrated implementation to eliminate redundancies. We utilize the Eigen library for matrix manipulations throughout data processing, including principal component extraction, codebook generation, and distance table creation. When the maximum number of neighbors $R$ exceeds the batch size $B$, we split each vertex's neighbor data into multiple blocks, each matching the batch size. This structure enables the efficient distance computation within each block using SIMD instructions.

\subsubsection{\textbf{Cost Analysis}}
\label{subsubsec: cost analysis}
In the CA stage, the time complexity of HNSW is approximately $O(R\cdot log(n))$, where $R$ is the maximum number of neighbors and $log(n)$ denotes the search path length \cite{graph_survey_vldb2021,HVS}. For each hop, vector data of the visiting vertices' neighbors must be accessed. Therefore, the number of memory accesses in the original HNSW algorithm ($NMA_{orig}$) can be expressed as:
\begin{equation}
    NMA_{orig} \sim O(R\cdot log(n))
\end{equation}
In our optimized implementation, neighbors' codewords and their respective IDs are stored contiguously, eliminating the need for random memory accesses to fetch neighbor vectors. As a result, the number of memory accesses in our implementation ($NMA_{ours}$) is:
\begin{equation}
    NMA_{ours} \sim O(log(n))
\end{equation}
For simplicity, this analysis excludes cache hits. Notably, the \texttt{Flash} code is significantly smaller than the original vector data, resulting in a higher cache hit ratio.
In the NS stage, computing distances between candidates requires accessing their vector data. For the original HNSW algorithm, the number of memory accesses is $O(C)$, where $C$ denotes the size of the candidate set. In contrast, \texttt{Flash} maintains the entire SDT in the L1 cache, avoiding fetching vector data from main memory during the NS stage.

For each distance computation using SIMD acceleration, the number of register loads in the original HNSW ($NRL_{orig}$) is:
\begin{equation}
    NRL_{orig} = \frac{32\cdot D}{U}
\end{equation}
where $D$ is the vector dimensionality (each vector has $D$ floating-point values) and $U$ is the bit-width of a register. This is because each vector segment must be individually loaded into an SIMD register. Furthermore, complex arithmetic operations, including subtraction, multiplication, and addition, are required for distance computation. In contrast, the number of register loads in our implementation ($NRL_{ours}$) is
\begin{equation}
    NRL_{ours} = \frac{M_{F}\cdot H}{U}
\end{equation}
where $M_F$ is the number of subspaces and $H$ represents the number of bits to encode a partial distance within each subspace. In \texttt{Flash}, the ADT is register-resident. Thus, it loads the codewords of $B$ neighbors within a subspace into a register in a single operation. To compute distances for these neighbors, $M_F$ batches of codewords are loaded, corresponding to $M_F$ register loads. The parameter $B$ is typically set to $U/H$ to align with register capacity. Additionally, the arithmetic operations in \texttt{Flash} are reduced to simple addition. With $D = 768$, $U = 128$, $M_F = 16$, and $H = 8$, the number of register loads required for a distance computation in the original HNSW is 192, whereas in our optimized version, it is reduced to just 1.
\setlength{\textfloatsep}{12pt plus 2pt minus 2pt}
\section{Evaluation}
\label{sec: experiments}
We conduct a comprehensive experimental evaluation to answer the following research questions:

\noindent\textbf{Q1}: What impact do mainstream compact coding algorithms have on HNSW construction? (Sections \ref{subsec: index perf} and \ref{subsec: search perf})

\noindent\textbf{Q2}: How does \texttt{Flash} affect HNSW's construction efficiency, index size, and search performance? (Sections \ref{subsec: index perf} and \ref{subsec: search perf})

\noindent\textbf{Q3}: How does \texttt{Flash} scale across varying data volumes and segment counts (Section \ref{subsec: scalability})?

\noindent{
\textbf{Q4}: How general is \texttt{Flash} with respect to various graph algorithms, HNSW optimizations, and SIMD instructions (Section \ref{subsec: generality})?
}

\noindent\textbf{Q5}: To what extent does \texttt{Flash} optimize memory accesses and arithmetic operations during index construction? (Section \ref{subsec: ablation})

\noindent\textbf{Q6}: How do the parameters of \texttt{Flash} influence construction efficiency and index quality? (Section \ref{subsec: param sensitivity})

All source codes and datasets are publicly available at: \url{https://github.com/ZJU-DAILY/HNSW-Flash}.

\subsection{Experimental Settings}

\begin{table}[t]
 \setstretch{0.9}
 \fontsize{7.5pt}{4mm}\selectfont
  \caption{Statistics of experimental datasets.}
  \vspace{-0.4cm}
  \label{tab:datasets}
  \setlength{\tabcolsep}{.02\linewidth}{
  \begin{tabular}{l|l|l|l|l}
    \hline
    \textbf{Datasets} & \textbf{Data Volume} & \textbf{Size (GB)} & \textbf{Dim.} & \textbf{Query Volume} \\
    \hline
    \hline
    ARGILLA \cite{argilla} & 21,071,228 & 81 & 1,024 & 100,000 \\
    \hline
    ANTON \cite{anton} & 19,399,177 & 75 & 1,024 & 100,000 \\
    \hline
    LAION \cite{laion} & 100,000,000 & 293 & 768 & 100,000 \\
    \hline
    IMAGENET \cite{imagenet} & 13,158,856 & 38 & 768 & 100,000 \\
    \hline
    COHERE \cite{cohere} & 10,124,929 & 30 & 768 & 100,000 \\
    \hline
    DATACOMP \cite{datacomp} & 12,800,000 & 37 & 768 & 100,000 \\
    \hline
    BIGCODE \cite{bigcode} & 10,404,628 & 30 & 768 & 100,000 \\
    \hline
    SSNPP \cite{bigann} & 1,000,000,000 & 960 & 256 & 100,000 \\
    \hline
  \end{tabular}
  }\vspace{-0.3cm}
\end{table}

\subsubsection{\textbf{Dataset}}
We use eight real-world vector datasets of diverse volumes and dimensions from deep embedding models, all of which are publicly accessible on Hugging Face and Big ANN Benchmarks websites \cite{huggingface, bigann}. A detailed summary of these datasets is presented in Table \ref{tab:datasets}. The query set is sampled from the dataset, with the corresponding ground truth generated through a linear scan.

\begin{figure*}
  \setlength{\abovecaptionskip}{0cm}
  \setlength{\belowcaptionskip}{0cm}
  \centering
  \footnotesize
  \hspace{2cm}
  \stackunder[0.5pt]{\includegraphics[scale=0.25]{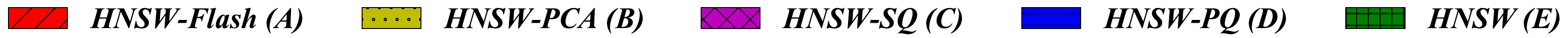}}{}
  \newline
  \stackunder[0.5pt]{\includegraphics[scale=0.27]{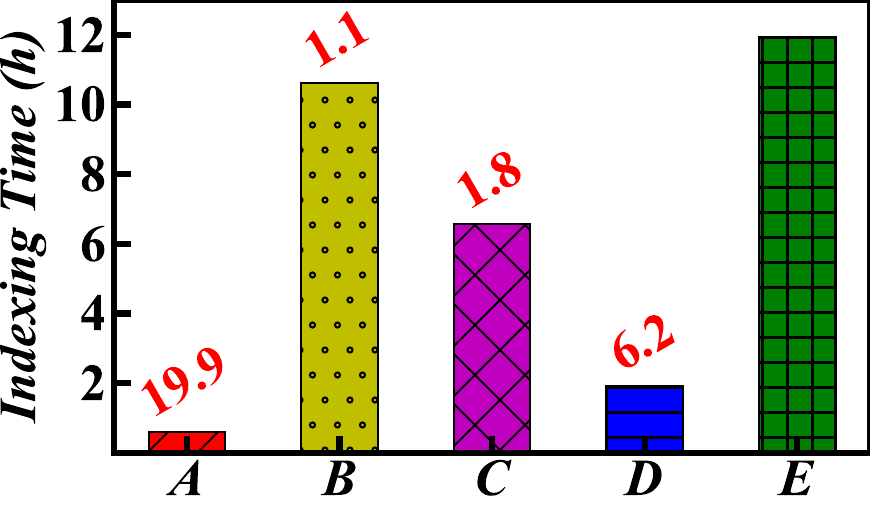}}{(a) SSNPP (10M)}
  \hspace{0.4cm}
  \stackunder[0.5pt]{\includegraphics[scale=0.27]{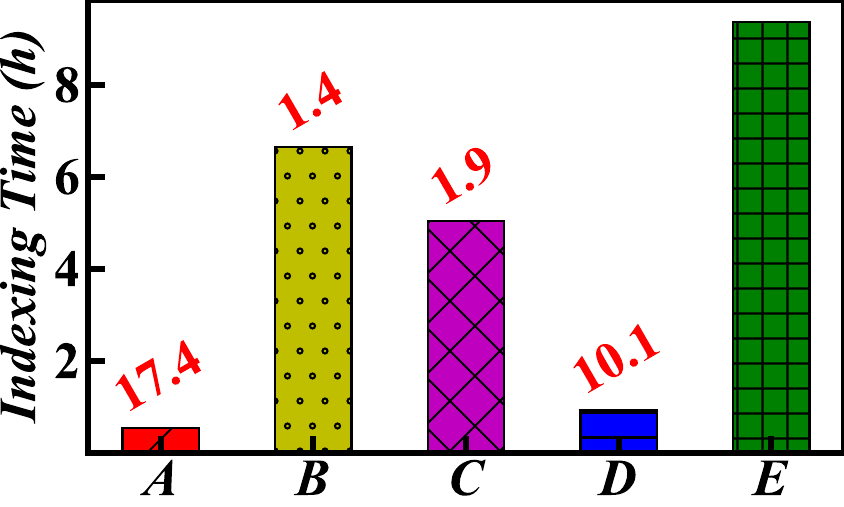}}{(b) LAION (10M)}
  \hspace{0.4cm}
  \stackunder[0.5pt]{\includegraphics[scale=0.27]{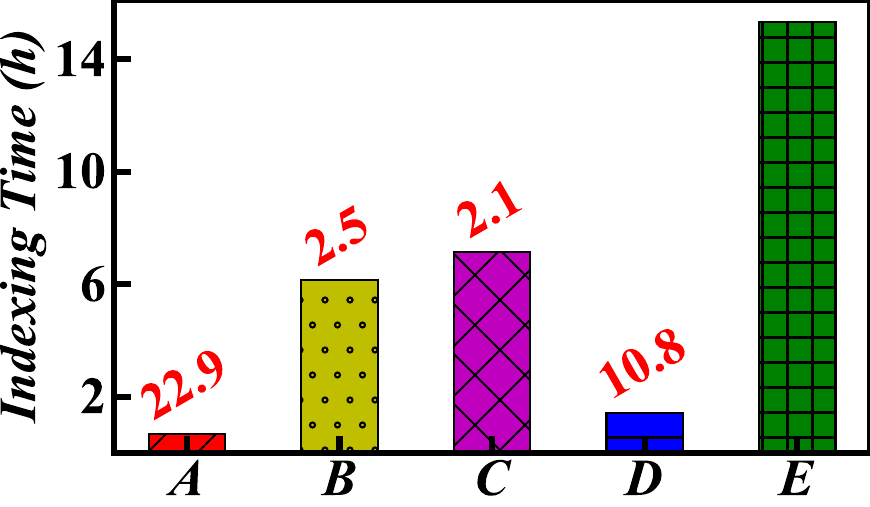}}{(c) COHERE (10M)}
  \hspace{0.4cm}
  \stackunder[0.5pt]{\includegraphics[scale=0.27]{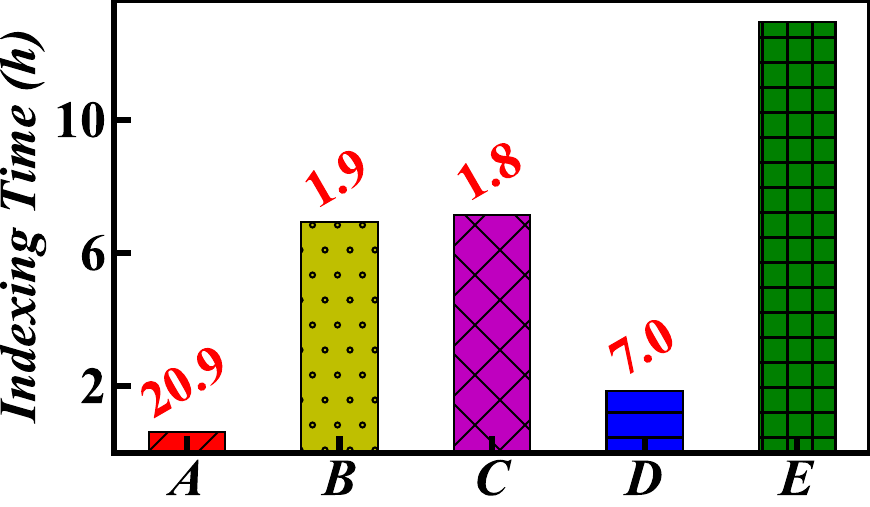}}{(d) BIGCODE (10M)}
  \newline
  \stackunder[0.5pt]{\includegraphics[scale=0.27]{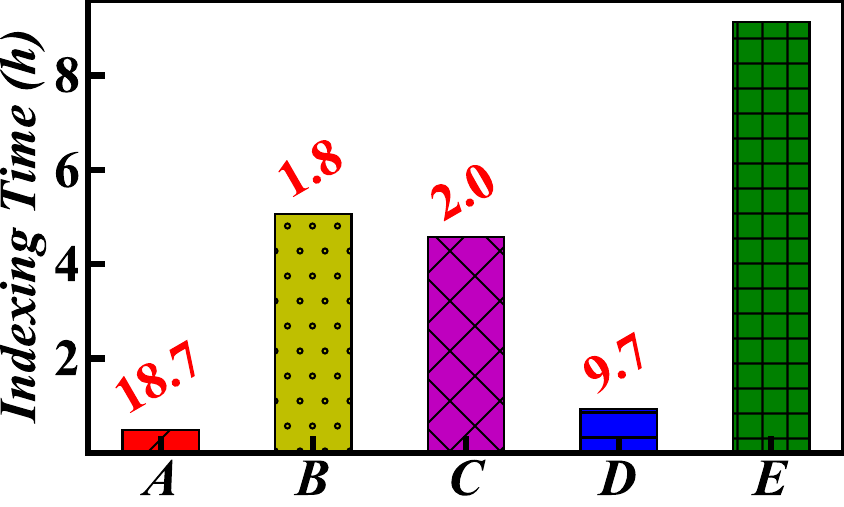}}{(e) IMAGENET (13M)}
  \hspace{0.4cm}
  \stackunder[0.5pt]{\includegraphics[scale=0.27]{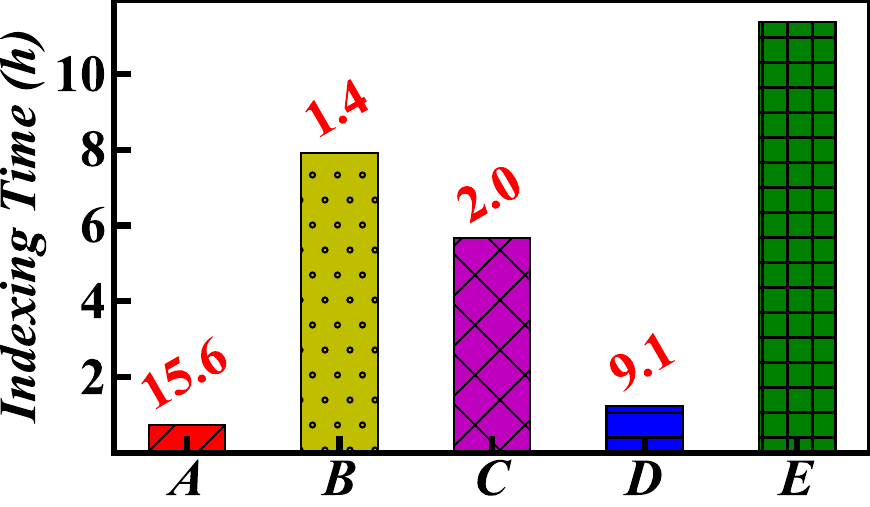}}{(f) DATACOMP (12M)}
  \hspace{0.4cm}
  \stackunder[0.5pt]{\includegraphics[scale=0.27]{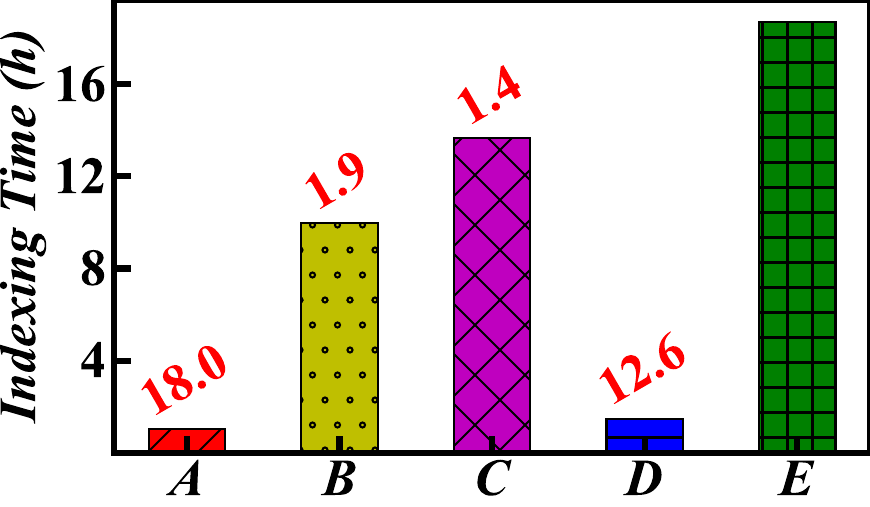}}{(g) ANTON (19M)}
  \hspace{0.4cm}
  \stackunder[0.5pt]{\includegraphics[scale=0.27]{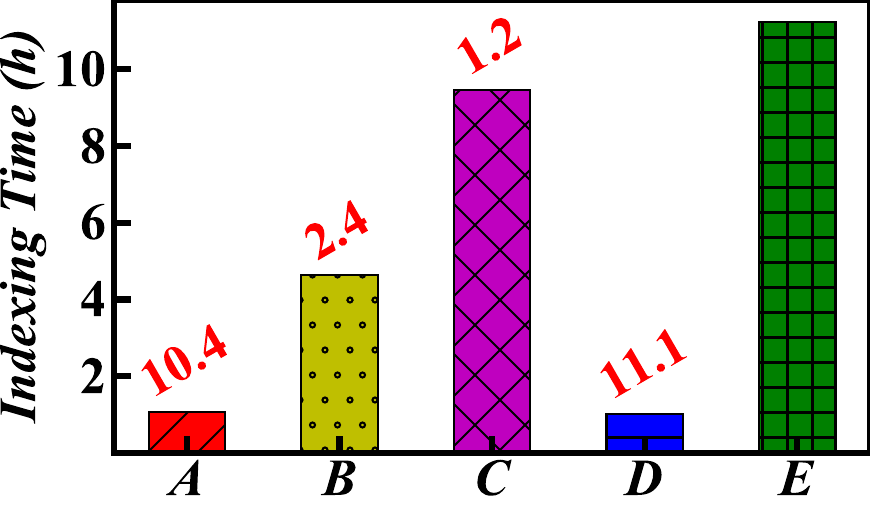}}{(h) ARGILLA (21M)}
  \newline
  \caption{Indexing times for all methods across eight datasets (red values at the top of each bar indicate speedup ratios).}
  \label{fig: index time}
  \vspace{-0.3cm}
\end{figure*}

\begin{figure*}
  \setlength{\abovecaptionskip}{0cm}
  \setlength{\belowcaptionskip}{0cm}
  \centering
  \footnotesize
  \hspace{2cm}
  \stackunder[0.5pt]{\includegraphics[scale=0.25]{figures/index_perf/indexing_perf_legend.pdf}}{}
  \newline
  \stackunder[0.5pt]{\includegraphics[scale=0.27]{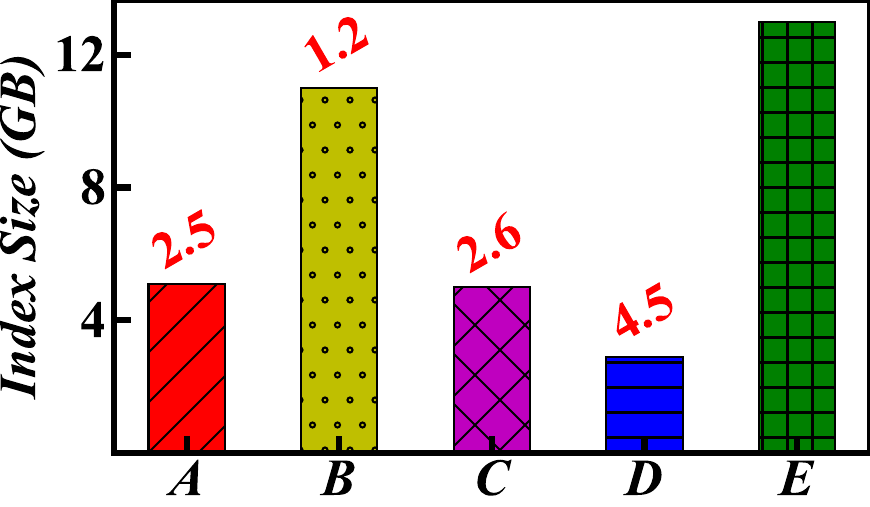}}{(a) SSNPP (10M)}
  \hspace{0.4cm}
  \stackunder[0.5pt]{\includegraphics[scale=0.27]{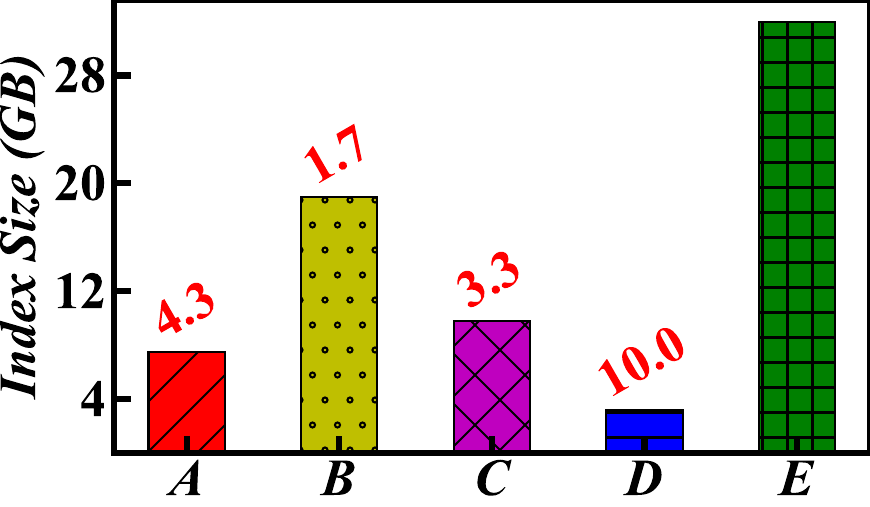}}{(b) LAION (10M)}
  \hspace{0.4cm}
  \stackunder[0.5pt]{\includegraphics[scale=0.27]{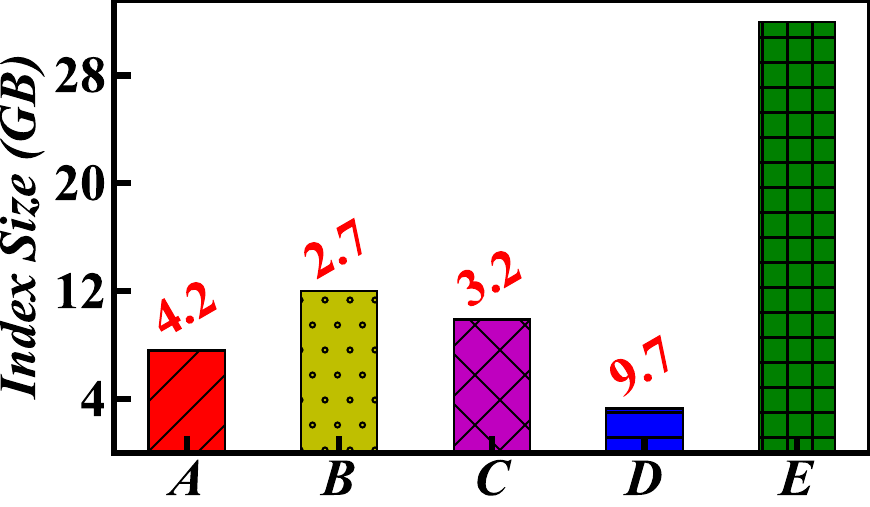}}{(c) COHERE (10M)}
  \hspace{0.4cm}
  \stackunder[0.5pt]{\includegraphics[scale=0.27]{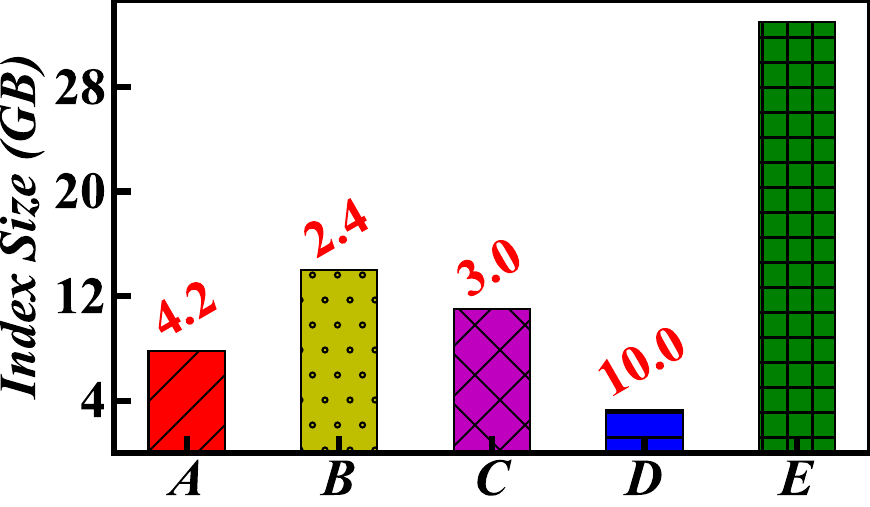}}{(d) BIGCODE (10M)}
  \newline
  \stackunder[0.5pt]{\includegraphics[scale=0.27]{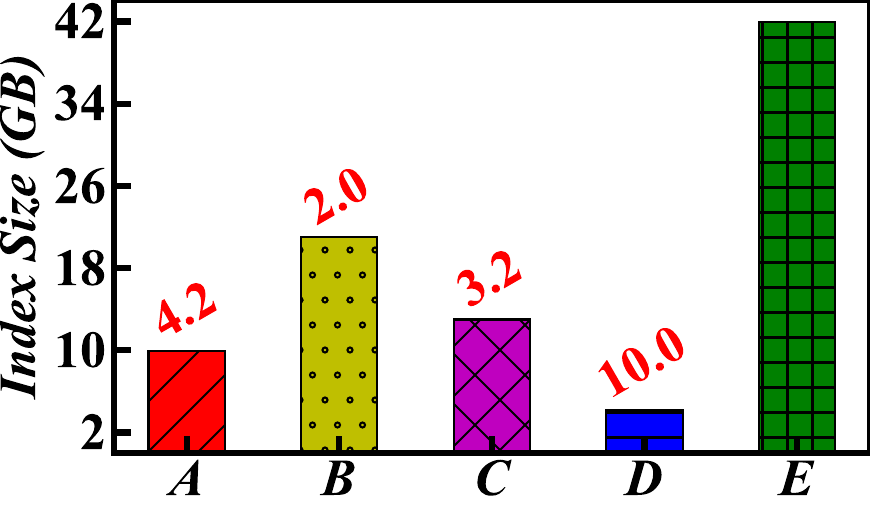}}{(e) IMAGENET (13M)}
  \hspace{0.4cm}
  \stackunder[0.5pt]{\includegraphics[scale=0.27]{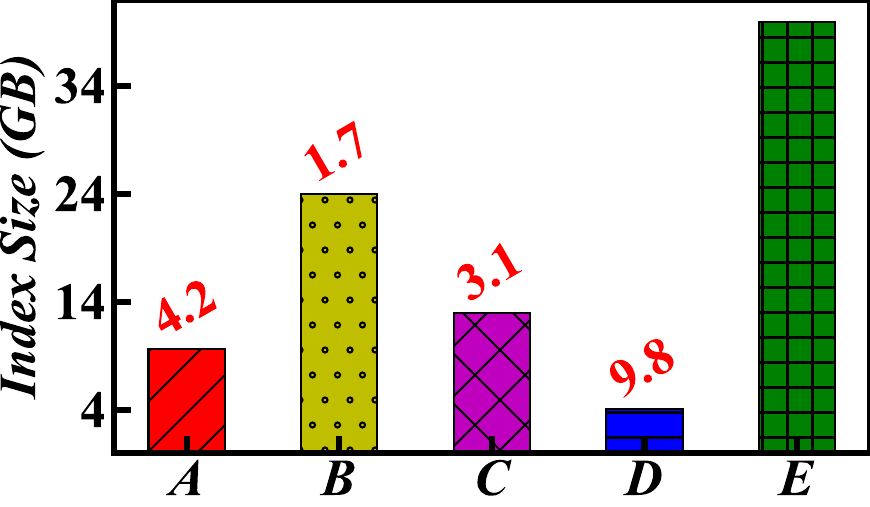}}{(f) DATACOMP (12M)}
  \hspace{0.4cm}
  \stackunder[0.5pt]{\includegraphics[scale=0.27]{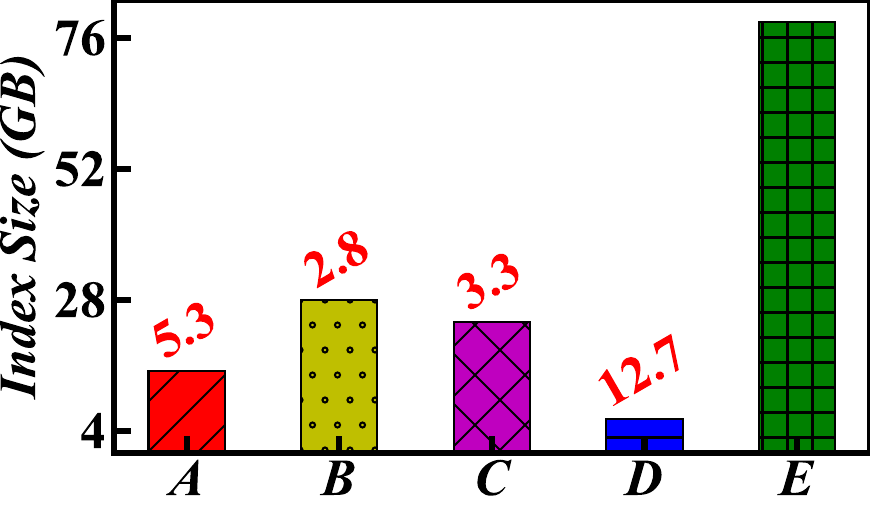}}{(g) ANTON (19M)}
  \hspace{0.4cm}
  \stackunder[0.5pt]{\includegraphics[scale=0.27]{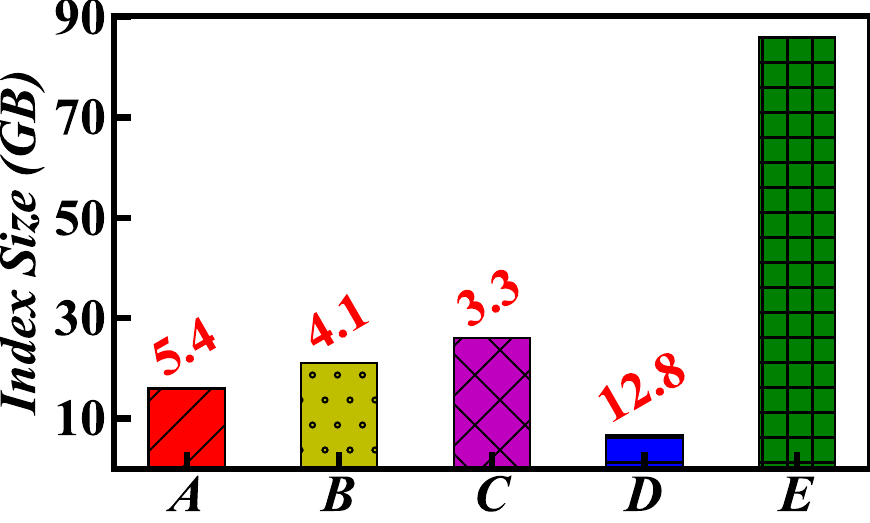}}{(h) ARGILLA (21M)}
  \newline
  \caption{Index sizes for all methods across eight datasets (red values at the top of each bar indicate compression ratios).}
  \label{fig: index size}
  \vspace{-0.3cm}
\end{figure*}

\subsubsection{\textbf{Compared Methods}}
We employ an advanced HNSW implementation from the latest \textit{hnswlib} repository \cite{hnswlib} as our benchmark, recognized as a standard in the field \cite{FosterK23, KhanAJ23}. This implementation is widely referenced by industrial vector databases \cite{SingleStore-V,hnsqlite,knowhere} and is a preferred choice for ANNS research \cite{ADSampling,PengZLJR23,IndykX23}. To accelerate index construction, three established compact coding methods are integrated into the repository, forming three baseline solutions. Additionally, the proposed \texttt{Flash} method is also incorporated into the repository for fair comparisons. {The generality of \texttt{Flash} is further evaluated on two recently optimized HNSW implementations (ADSampling \cite{ADSampling} and VBase \cite{zhang2023vbase}) and two additional graph algorithms (NSG \cite{NSG} and $\tau$-MG \cite{tau-MG}).}

\noindent$\bullet$ \textbf{HNSW} \cite{HNSW}. HNSW is a prominent graph-based ANNS algorithm \cite{graph_survey_vldb2021, DPG, ADSampling}. Its prevalence in industrial applications \cite{zhang2023vbase,PASE,ADBV} and extensive academic research \cite{LiZAH20,ZhangLW24,ZuoQZLD24} underscore its significance (see Section \ref{subsec: Related Work}). Given its representativeness, HNSW serves as a key graph index example in this research.

\noindent$\bullet$ \textbf{HNSW-PQ}. Product Quantization (PQ) is a recognized vector compression method in high-dimensional vector search \cite{PQ,OPQ,ScaNN}. We integrate PQ into the HNSW index construction process to accelerate graph indexing (see Section \ref{subsubsec: HNSW PQ} for details).

\noindent$\bullet$ \textbf{HNSW-SQ}. Scalar Quantization (SQ) converts floating-point values to integers, reducing memory overhead and enhancing search performance \cite{LVQ,LeanVec,qdrant-sq}. We integrate it into HNSW construction according to Section \ref{subsubsec: HNSW SQ}, and implement an optimized version to avoid decoding overhead based on a recent technical report \cite{qdrant-sq}.

\noindent$\bullet$ \textbf{HNSW-PCA}. Principal Component Analysis (PCA) is a classical dimensionality reduction method that enhances calculation efficiency by reducing the need to scan all dimensions \cite{ADSampling,yang2024bridging}. We integrate PCA into the construction process following Section \ref{subsubsec: HNSW PCA}.

\noindent$\bullet$ \textbf{HNSW-Flash}. We incorporate the proposed compact coding method, \texttt{Flash}, into the HNSW construction process. This method thoughtfully considers the unique characteristics of HNSW construction, optimizing memory layout and arithmetic operations to minimize random memory accesses and enhance SIMD utilization.

\subsubsection{\textbf{Parameter Settings}}
\label{subsubsec: param set}
For all methods compared, we set the maximum number of candidates ($C$) to 1024 and the maximum number of neighbors ($R$) to 32, following prior research recommendations\cite{HVS}. For compact coding parameters, we use established conventions to determine optimal settings. Specifically, the dimension of principal components ($d_{PCA}$) is chosen to ensure accumulated variance exceeds 0.9 in HNSW-PCA, as guided by our evaluations (see Section \ref{subsubsec: HNSW PCA}). For HNSW-SQ, we set the number of bits ($L_{SQ}$) to 8, which has been demonstrated as optimal in our experiments and existing literature \cite{qdrant-sq,LeanVec}. In HNSW-PQ, the number of subspaces ($M_{PQ}$) is identified through grid search, with bits per subspace ($L_{PQ}$) set to 8, consistent with common configurations in the literature \cite{PQ,OPQ,ZhanM0GZM21,PQfast}. For HNSW-Flash, we maintain the same subspace count ($M_{F}$) as HNSW-PQ; and the dimension of principal components ($d_{F}$) is selected via grid search. The bits per subspace ($L_{F}$) are set to 4, while each distance value is allocated 8 bits ($H=8$), enabling the ADT within a subspace to fit entirely within a SIMD register. Additionally, the batch size ($B$) for organizing the neighbor list is set to 16 to adapt to the number of distances within an ADT ($B=K=2^{L_F}$) and optimize register loading, resulting in each neighbor list comprising two blocks for separate execution of distance computations ($R/B=2$).

\subsubsection{\textbf{Performance Metrics}}
\label{subsubsec: perf metrics}
To evaluate index construction efficiency, we record indexing time, which includes data preprocessing times for algorithms with coding techniques. We assess index quality by measuring both search efficiency and accuracy for the query set. Search efficiency is quantified in Queries Per Second ($QPS$), representing the number of queries processed per second. {Search accuracy is measured through \textit{Recall} and the \textit{average distance ratio} (\textit{ADR}).} \textit{Recall} is defined as $Recall=\frac{|G\cup S|}{k}$, where $G$ is the ground truth set of the $k$ nearest vectors, and $S$ denotes the set of $k$ results returned by the algorithm, with $k$ set to 1 by default. {\textit{ADR}, defined in \cite{PatellaC08,TaoYSK10}, represents the average of the distance ratios between the retrieved $k$ database vectors and the ground truth nearest vectors.} All experimental results are derived from three repeated executions.

\subsubsection{\textbf{Environment Configuration}}
The experiments are conducted on a Linux server with dual Intel Xeon E5-2620 v3 CPUs (2.40GHz, 24 logical processors). The server employs an x86\_64 architecture with multiple cache levels (32KB L1, 256KB L2, and 15MB L3) and 378GB RAM. All methods are implemented in C++ using intrinsics for SIMD access, compiled with g++ 9.3.1 using the \textit{-Ofast} and \textit{-march=native} options. {SSE instructions are used by default for SIMD acceleration, and we also compare SSE (128-bit), AVX (256-bit), and AVX512 (512-bit)\footnote{{A new server is used for this evaluation due to the default does not support AVX512.}}.} The Eigen library \cite{eigen} optimizes matrix manipulations, while OpenMP facilitates parallel index construction and search execution across 24 threads.

\subsection{Indexing Performance}
\label{subsec: index perf}
The indexing times and index sizes for different methods are presented in Figures \ref{fig: index time} and \ref{fig: index size}, respectively. The indexing time for HNSW reflects only the graph index construction, while other methods include extra coding preprocessing. The results show that optimized methods exhibit shorter indexing times than the original HNSW, indicating that compact coding techniques effectively accelerate HNSW construction. HNSW-Flash consistently outperforms other methods across almost all datasets, achieving speedups of 10.4$\times$ to 22.9$\times$, highlighting \texttt{Flash}'s superiority in expediting index construction. Among baseline methods, HNSW-PCA and HNSW-SQ yield less than 2$\times$ speedup on most datasets, whereas HNSW-PQ demonstrates a better speedup ratio.
Figure \ref{fig: index size} illustrates that all optimized methods decrease the index size compared to the original HNSW by replacing high-dimensional vectors with compact vector codes. HNSW-PQ achieves the highest compression ratio, explaining its superior speedup in index construction. HNSW-Flash’s memory layout optimization, which stores neighbor codewords alongside IDs, results in a lower compression ratio than HNSW-PQ; however, it reaches a better speedup due to efficient memory access and arithmetic operations. Notably, HNSW-PQ suffers from a significant reduction in search performance (Figure \ref{fig: search perf}) due to its high compression error, indicating the poor index quality.

\begin{figure*}
  \setlength{\abovecaptionskip}{0cm}
  \setlength{\belowcaptionskip}{0cm}
  \centering
  \footnotesize
  \hspace{2cm}
  \stackunder[0.5pt]{\includegraphics[scale=0.25]{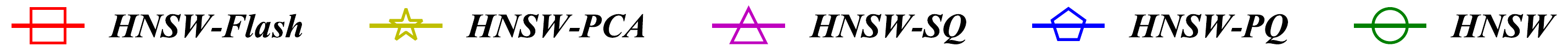}}{}
  \newline
  \stackunder[0.5pt]{\includegraphics[scale=0.26]{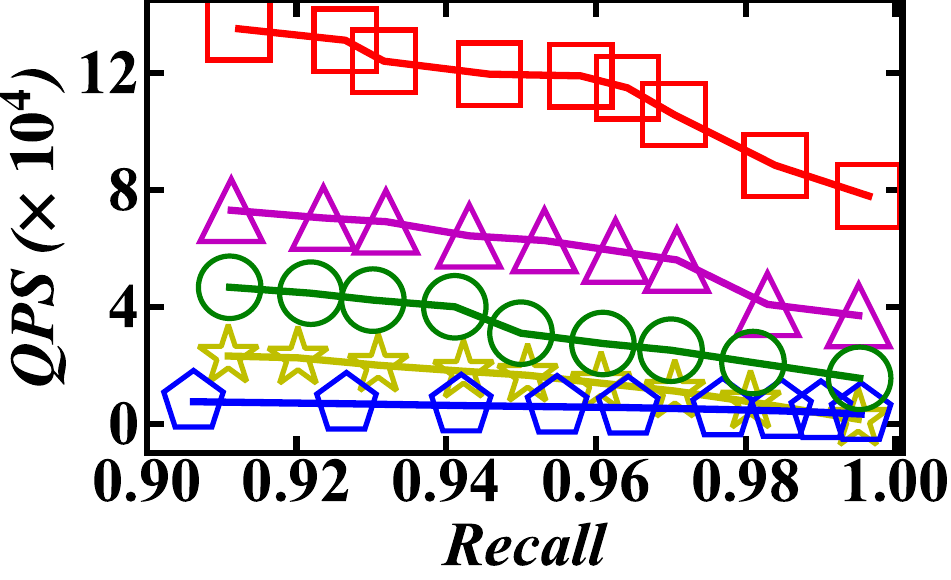}}{(a) SSNPP (10M)}
  \hspace{0.29cm}
  \stackunder[0.5pt]{\includegraphics[scale=0.26]{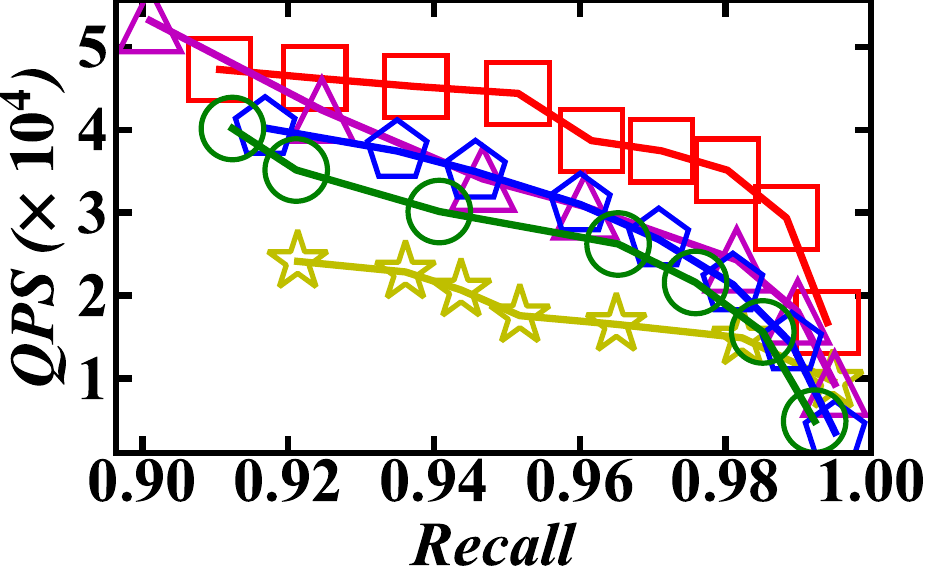}}{(b) LAION (10M)}
  \hspace{0.26cm}
  \stackunder[0.5pt]{\includegraphics[scale=0.26]{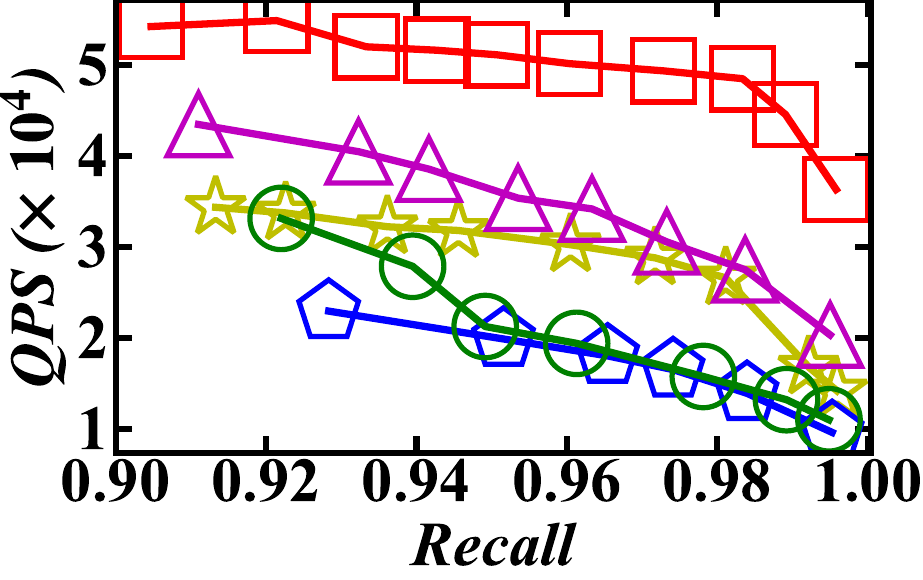}}{(c) COHERE (10M)}
  \hspace{0.26cm}
  \stackunder[0.5pt]{\includegraphics[scale=0.26]{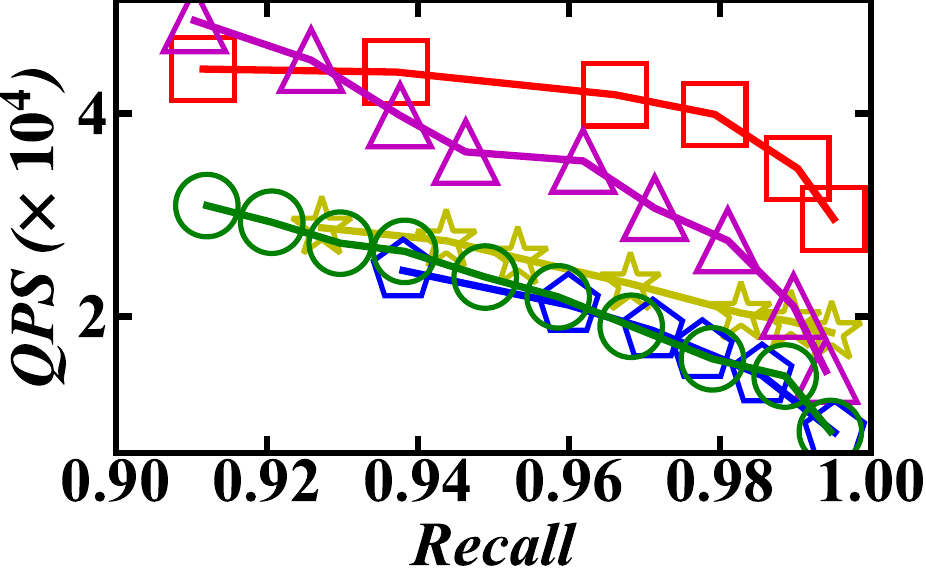}}{(d) BIGCODE (10M)}
  \newline
  \stackunder[0.5pt]{\includegraphics[scale=0.26]{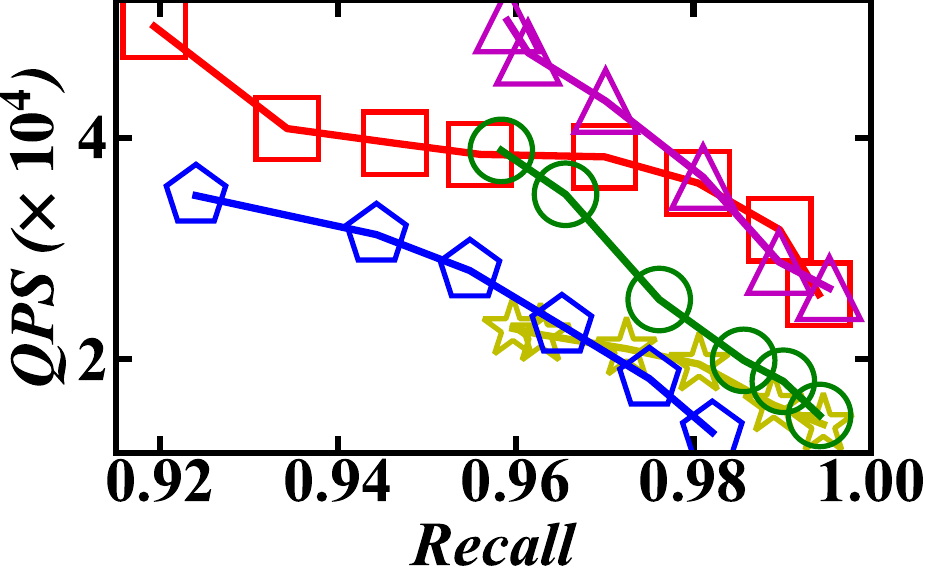}}{(e) IMAGENET (13M)}
  \hspace{0.25cm}
  \stackunder[0.5pt]{\includegraphics[scale=0.26]{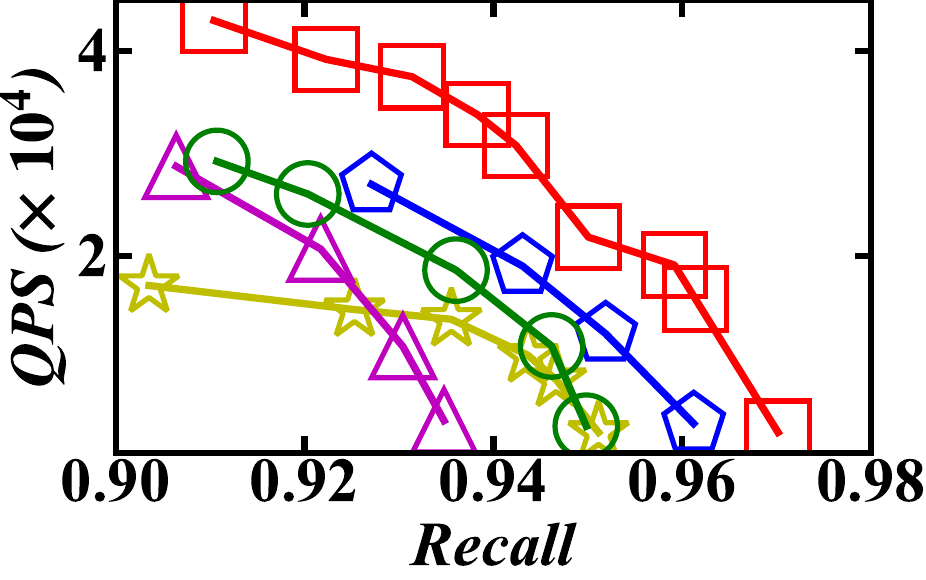}}{(f) DATACOMP (12M)}
  \hspace{0.22cm}
  \stackunder[0.5pt]{\includegraphics[scale=0.26]{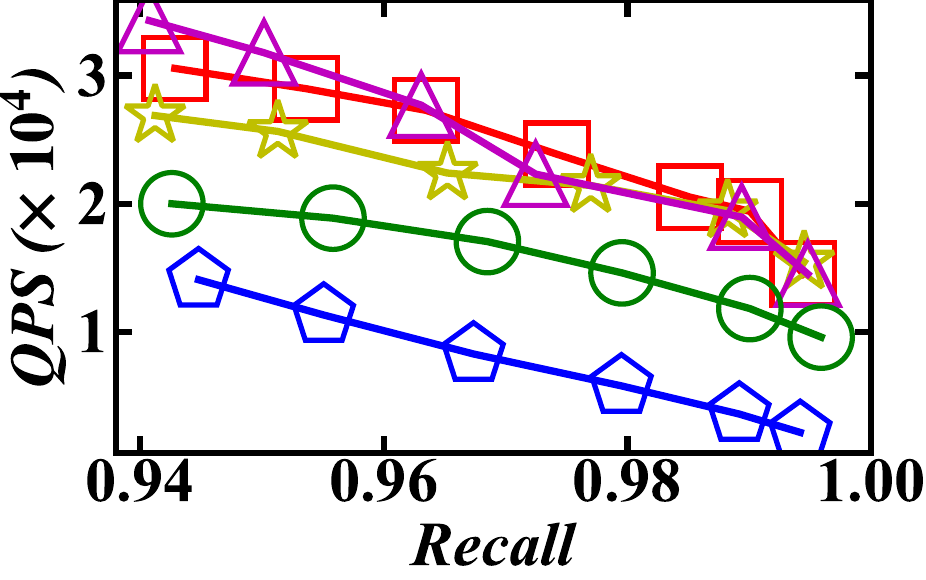}}{(g) ANTON (19M)}
  \hspace{0.35cm}
  \stackunder[0.5pt]{\includegraphics[scale=0.26]{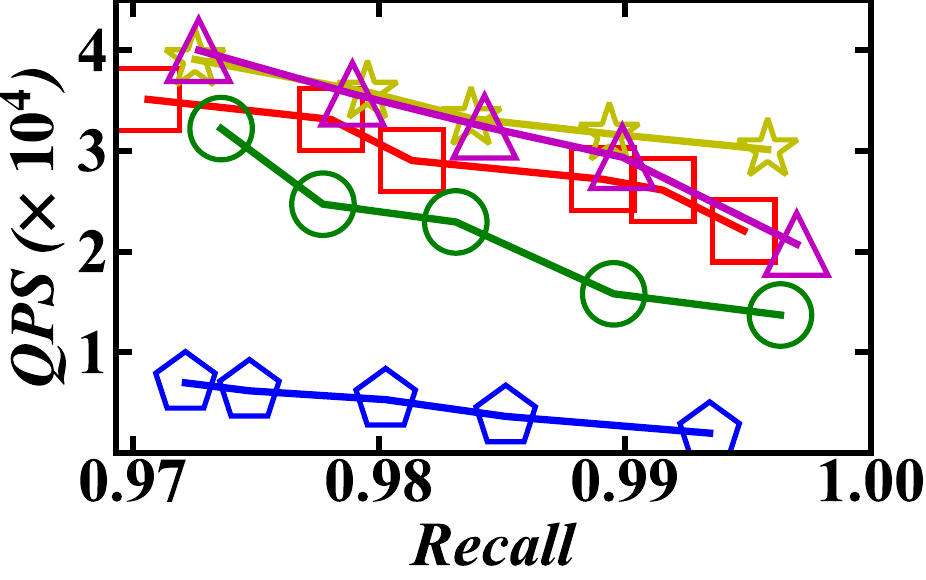}}{(h) ARGILLA (21M)}
  \newline
  \caption{Comparison of search performance across eight datasets (the top right indicates better performance).}
  \label{fig: search perf}
  \vspace{-0.3cm}
\end{figure*}

\begin{figure}
\vspace{-0.1cm}
  \setlength{\abovecaptionskip}{0cm}
  \setlength{\belowcaptionskip}{0cm}
  \centering
  \footnotesize
  \stackunder[0.5pt]{\includegraphics[scale=0.17]{figures/search_perf/search_perf_legend.pdf}}{}
  \newline
  \stackunder[0.5pt]{\includegraphics[scale=0.25]{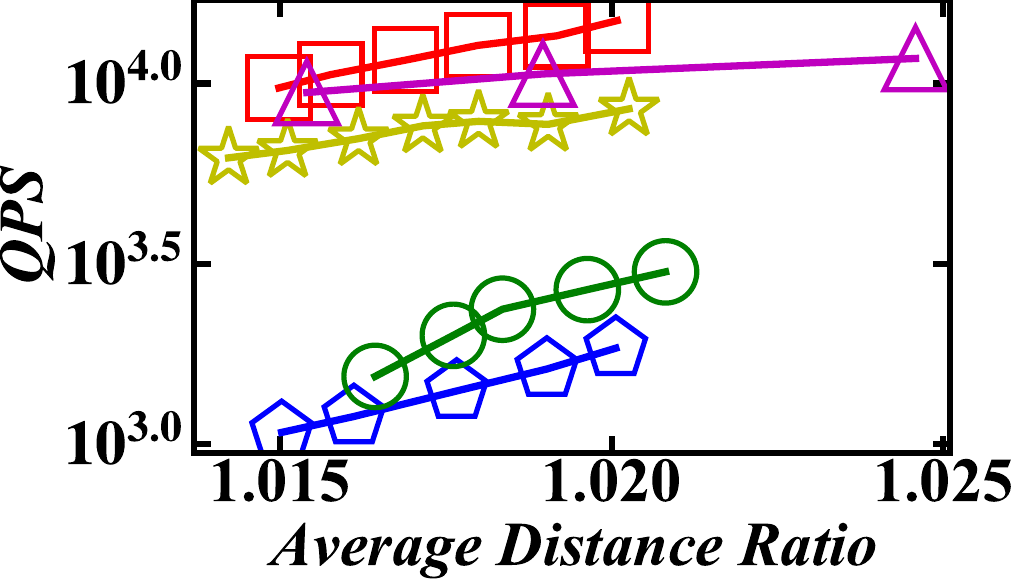}}{(a) LAION (10M)}
  \stackunder[0.5pt]{\includegraphics[scale=0.25]{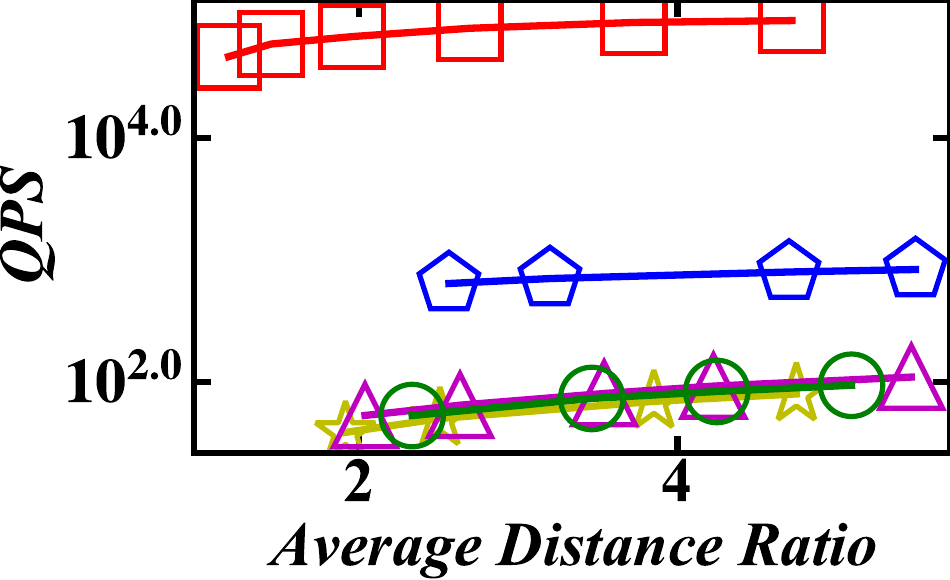}}{(b) SSNPP (10M)}
  \newline
  \caption{{The \textit{QPS}-\textit{ADR} curves (the top left is better).}}
  \label{fig: search adr}
  \vspace{-0.2cm}
\end{figure}

\subsection{Search Performance}
\label{subsec: search perf}

{Figure \ref{fig: search perf} presents the \textit{QPS}-\textit{Recall} curves for all evaluated datasets, while Figure \ref{fig: search adr} shows the \textit{QPS}-\textit{ADR} curves for two representative datasets, with analogous trends observed for the remaining datasets.} The results indicate that HNSW-Flash consistently outperforms other methods across nearly all datasets, demonstrating its ability to accelerate index construction without compromising search performance. Additionally, the optimizations in the Candidate Acquisition (CA) stage can be seamlessly integrated into the search procedure, as it shares similar steps with the CA process.
The search performance of baseline methods varies significantly across datasets. For instance, HNSW-PCA achieves the best performance on ARGILLA but the worst on LAION, indicating sensitivity to data distribution. While HNSW-PQ offers considerable speedup in index construction, its search performance is inferior to standard HNSW due to degraded index quality. HNSW-SQ shows consistent search performance gains on most datasets, aligning with prior studies that optimizes search procedure using SQ \cite{LVQ}. However, HNSW-SQ provides limited indexing speedup, suggesting that a compact coding method suitable for search may not be ideal for index construction. This occurs because index construction involves an additional Neighbor Selection (NS) step and is therefore more complex than the search process. {In addition, different methods may exhibit varying search performance rankings regarding \textit{Recall} and \textit{ADR} for the same dataset, indicating their false positive results are obviously different. Notably, \texttt{Flash} demonstrates substantial advantages in retrieving results closer to the ground truth vectors.}

\begin{figure}
\vspace{-0.1cm}
  \setlength{\abovecaptionskip}{0cm}
  \setlength{\belowcaptionskip}{0cm}
  \centering
  \footnotesize
  \stackunder[0.5pt]{\includegraphics[scale=0.26]{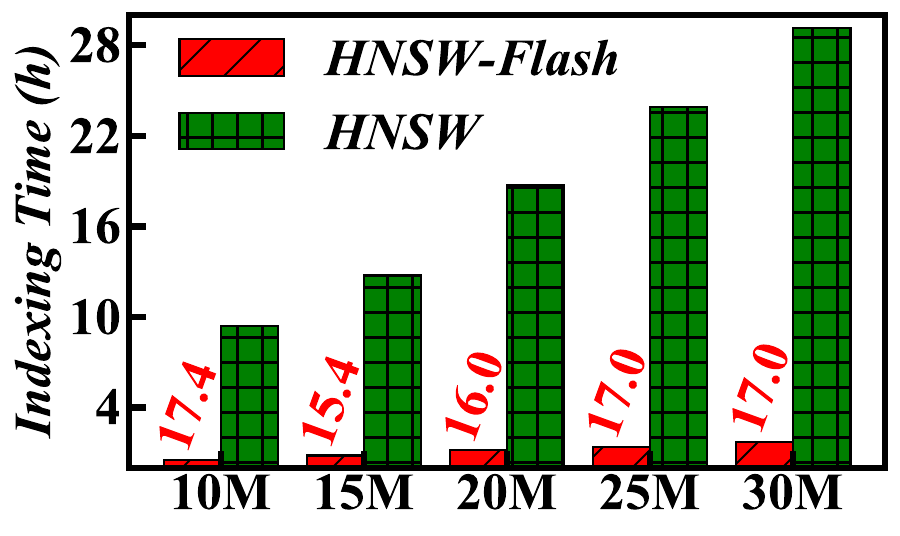}}{(a) LAION-30M}
  \hspace{0.15cm}
  \stackunder[0.5pt]{\includegraphics[scale=0.26]{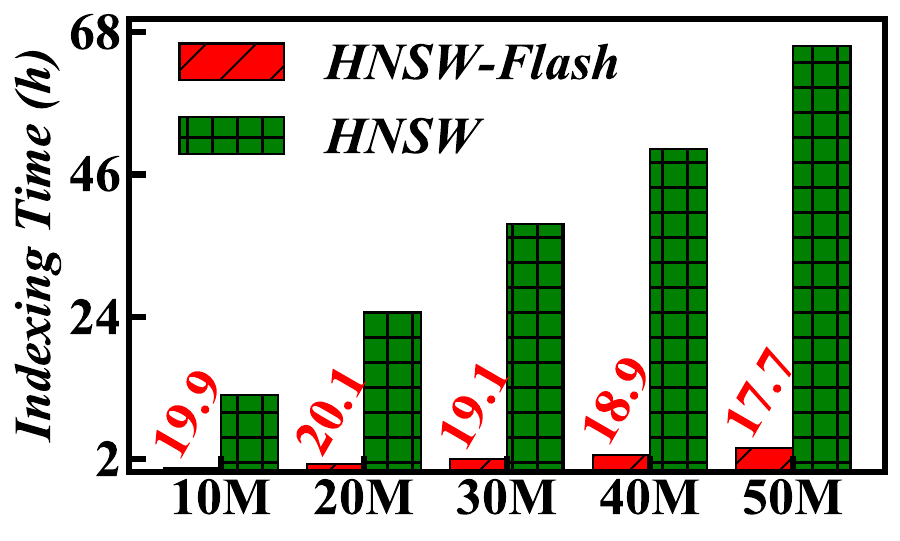}}{(b) SSNPP-50M}
  \newline
  \caption{Scalability over different data volumes.}
  \label{fig: scalability data volume}
  \vspace{-0.2cm}
\end{figure}

\subsection{Scalability}
\label{subsec: scalability}
We evaluate the scalability of HNSW-Flash regarding data volume and the number of data segments, benchmarking its performance against standard HNSW. Indexing times for both HNSW and HNSW-Flash are assessed on the LAION-100M and SSNPP-1B datasets. Figure \ref{fig: scalability data volume} presents indexing times across varying data volumes, with the data size within a single segment progressively increasing, while Figure \ref{fig: scalability segments} depicts indexing times for different segment counts, maintaining consistent data size per segment. For multi-segment evaluation, we accumulate the indexing times of each segment. The results demonstrate that HNSW-Flash significantly outperforms HNSW across all data volumes and segment counts. This highlights HNSW-Flash's superiority, particularly with larger datasets where HNSW incurs considerably higher indexing times. Consequently, HNSW-Flash achieves a notable reduction in absolute indexing times. In addition, HNSW-Flash can be seamlessly integrated into current distributed systems to expedite index construction within each segment, thereby providing cumulative acceleration.

\begin{figure}
  \setlength{\abovecaptionskip}{0cm}
  \setlength{\belowcaptionskip}{0cm}
  \centering
  \footnotesize
  \stackunder[0.5pt]{\includegraphics[scale=0.26]{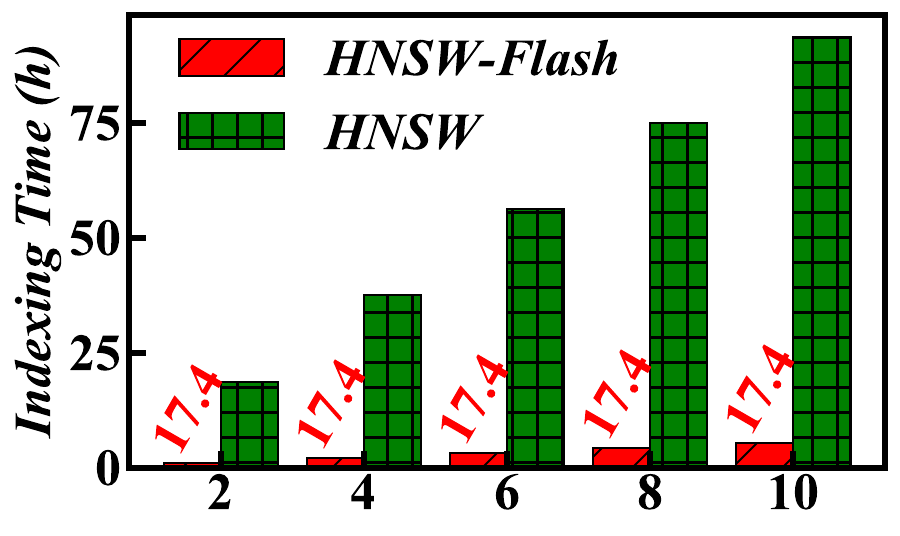}}{(a) LAION-100M}
  \hspace{0.15cm}
  \stackunder[0.5pt]{\includegraphics[scale=0.26]{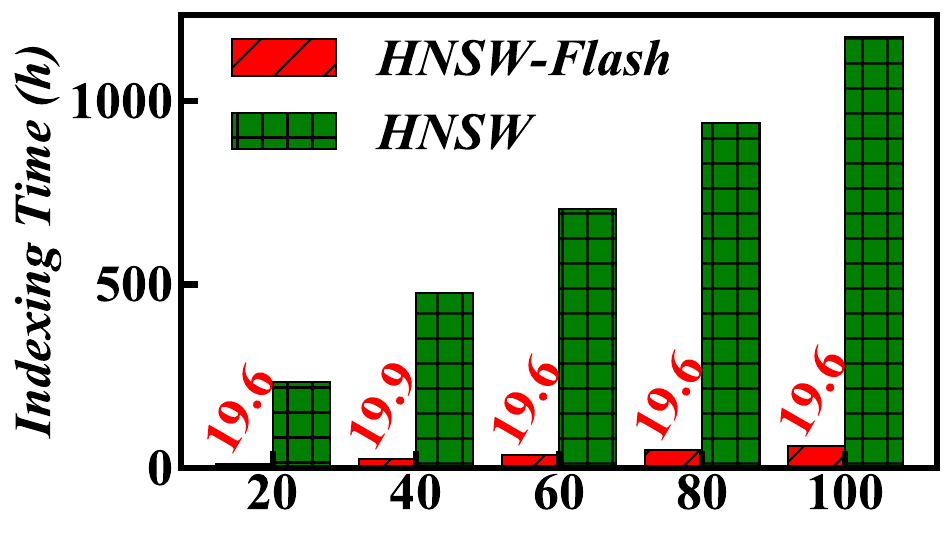}}{(b) SSNPP-1B}
  \newline
  \caption{Scalability over different segment counts.}
  \label{fig: scalability segments}
  \vspace{-0.2cm}
\end{figure}

\begin{figure}
\vspace{-0.2cm}
  \setlength{\abovecaptionskip}{0cm}
  \setlength{\belowcaptionskip}{0cm}
  \centering
  \footnotesize
  \stackunder[0.5pt]{\includegraphics[scale=0.26]{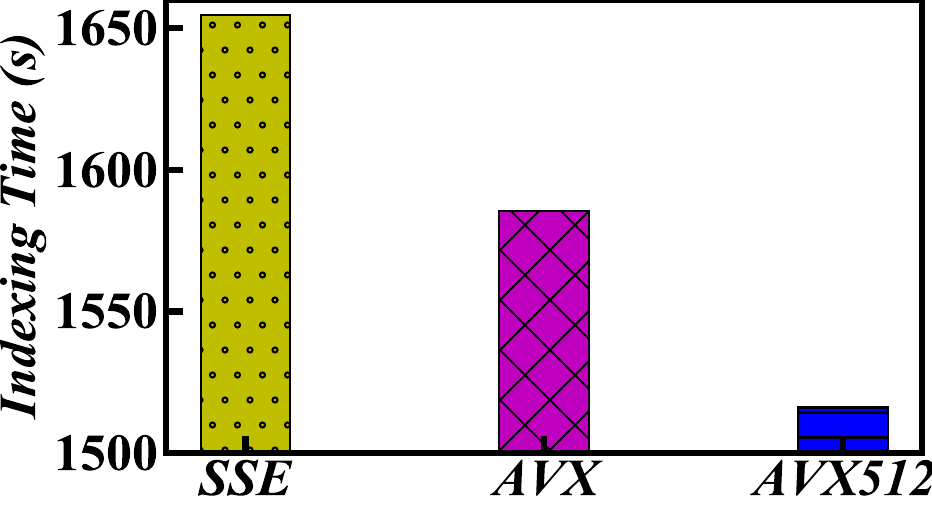}}{(a) LAION (10M)}
  \hspace{0.15cm}
  \stackunder[0.5pt]{\includegraphics[scale=0.26]{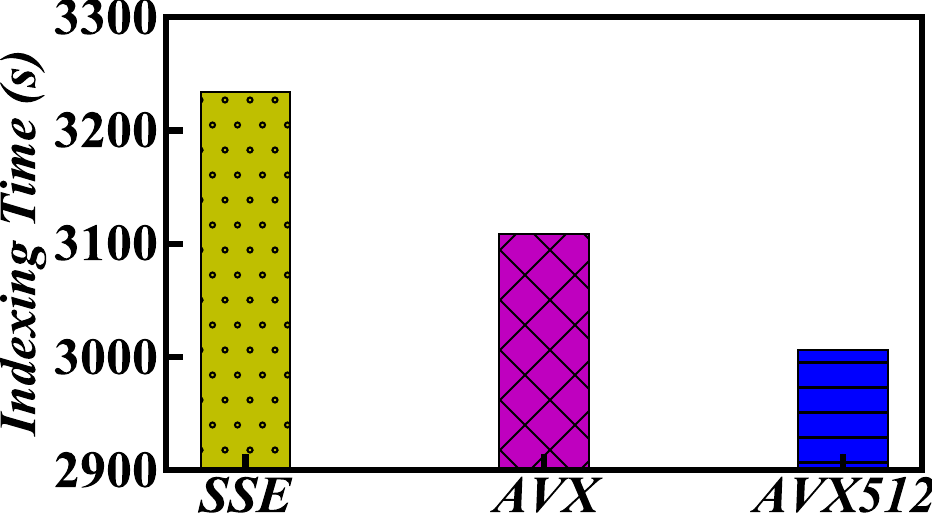}}{(b) SSNPP (10M)}
  \newline
  \caption{{Generality across different SIMD instruction sets.}}
  \label{fig: generality simd}
  \vspace{-0.2cm}
\end{figure}

\subsection{Generality}
\label{subsec: generality}
\subsubsection{{\textbf{SIMD Instructions}}}
\label{subsubsec: simd instructions}
{We integrate \texttt{Flash} into various SIMD instruction sets with differing register sizes to confirm its generality. We highlight \texttt{Flash} can be extended to 256-bit and 512-bit registers without altering its fundamental design principle. As shown in Figure \ref{fig: generality simd}, more advanced SIMD instructions with larger registers exhibit faster indexing processes. This is because SIMD instructions with larger registers can process more asymmetric distance tables (ADTs) per operation, resulting in more efficient arithmetic operations compared to those with smaller registers. Notably, the acceleration gained from larger registers is not strictly linear. Other factors, such as memory access patterns, also significantly impact indexing time. Moreover, the specific latency and throughput of instructions (such as register load times) vary across different SIMD instruction sets \cite{Intel-SIMD}, further affecting the overall speedup.}

\begin{figure}
\vspace{-0.2cm}
  \setlength{\abovecaptionskip}{0cm}
  \setlength{\belowcaptionskip}{0cm}
  \centering
  \footnotesize
  \stackunder[0.5pt]{\includegraphics[scale=0.26]{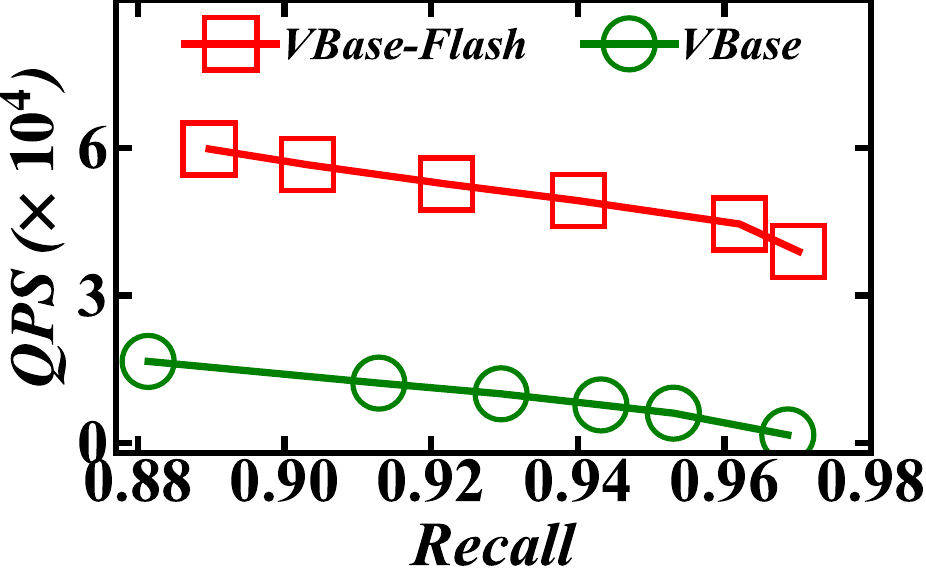}}{(a) VBase}
  \hspace{0.15cm}
  \stackunder[0.5pt]{\includegraphics[scale=0.26]{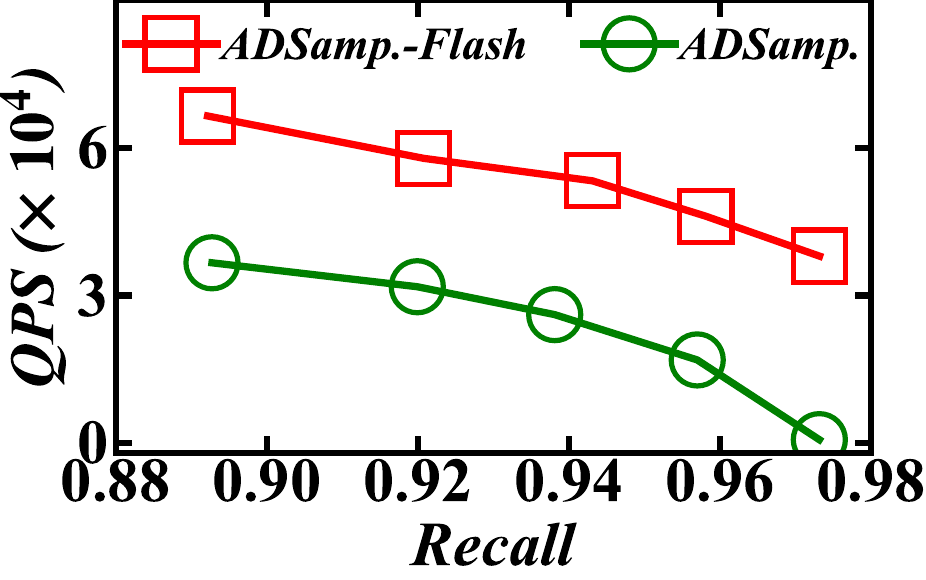}}{(b) ADSampling}
  \newline
  \caption{{Generality across different implementations.}}
  \label{fig: generality impl}
  \vspace{-0.2cm}
\end{figure}

\subsubsection{\textbf{{Optimized HNSW Implementations}}}
\label{subsubsec: opt hnsw impl}
{We apply \texttt{Flash} to optimized HNSW implementations, ADSampling \cite{ADSampling} and VBase \cite{zhang2023vbase}. These optimizations retain the standard HNSW construction process, enabling \texttt{Flash} to directly accelerate indexing, as shown in Figure \ref{fig: index time}. Figure \ref{fig: generality impl} presents the search performance of these implementations on LAION-10M before and after integrating \texttt{Flash}, demonstrating that \texttt{Flash} further improves search efficiency. This improvement stems from the fact that ADSampling’s enhancements to distance comparisons and VBase’s adjustments to the termination condition are orthogonal to \texttt{Flash}.}

\begin{figure}
  \setlength{\abovecaptionskip}{0cm}
  \setlength{\belowcaptionskip}{0cm}
  \centering
  \footnotesize
  \stackunder[0.5pt]{\includegraphics[scale=0.25]{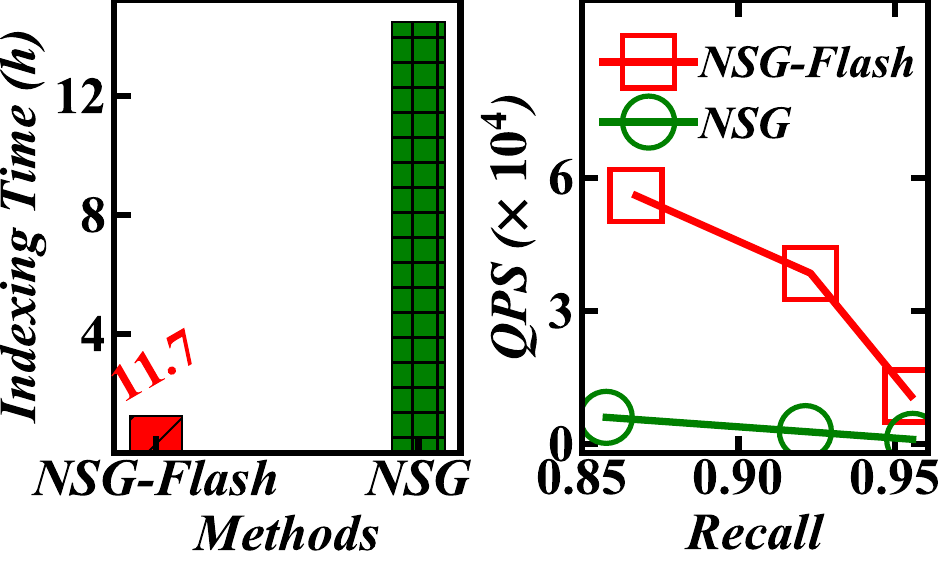}}{(a) NSG}
  \hspace{0.15cm}
  \stackunder[0.5pt]{\includegraphics[scale=0.25]{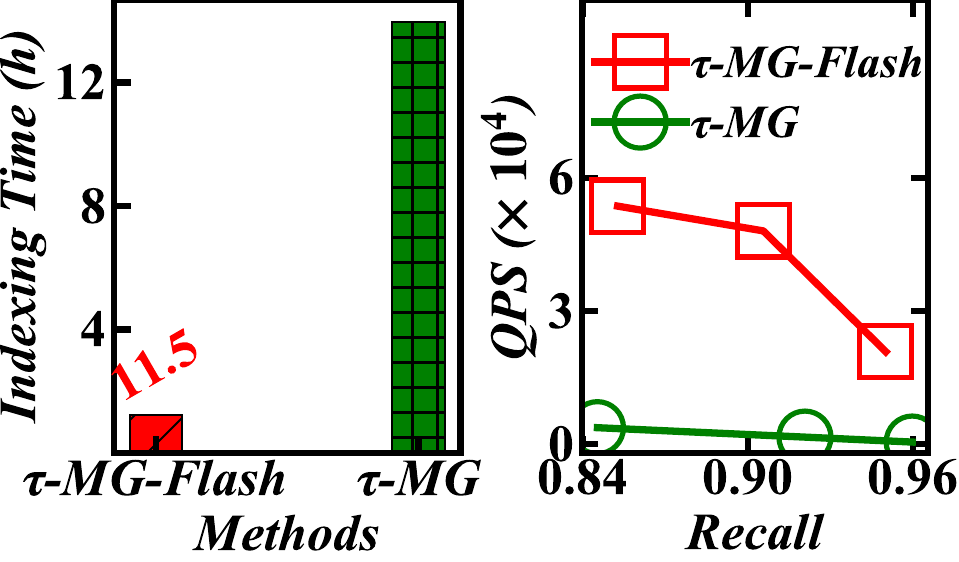}}{(b) $\tau$-MG}
  \newline
  \caption{{Generality across different graph algorithms.}}
  \label{fig: generality graph}
  \vspace{-0.2cm}
\end{figure}

\subsubsection{\textbf{{Graph Algorithms}}}
\label{subsubsec: graph algorithms}
{We apply \texttt{Flash} to two other representative graph algorithms, NSG \cite{NSG} and $\tau$-MG \cite{tau-MG}. Figure \ref{fig: generality graph} presents the indexing time and search performance on LAION-10M. The results indicate that \texttt{Flash} significantly accelerates the indexing process of NSG and $\tau$-MG while improving their search performance. These findings are consistent with those observed in HNSW. Both algorithms and HNSW share two key components—Candidate Acquisition (CA) and Neighbor Selection (NS)—in their indexing processes. Consequently, enhancing the CA and NS steps accelerates the indexing procedure across these graph algorithms.}

\subsection{Ablation Study}
\label{subsec: ablation}

\begin{table}[t]
\vspace{-0.2cm}
 \setstretch{0.9}
 \fontsize{7.5pt}{4mm}\selectfont
  \caption{L1 cache misses before and after optimization.}
  \vspace{-0.4cm}
  \label{tab:cache misses}
  \setlength{\tabcolsep}{.007\linewidth}{
  \begin{tabular}{l|l|l|l|l|l|l|l|l}
    \hline
     & SSNPP & LAION & COHE. & BIGC. & IMAGE. & DATAC. & ANTON & ARGIL. \\
    \hline
    \hline
     \textbf{w/o opt.} & 19.08\% & 24.20\% & 24.97\% & 25.09\% & 23.02\% & 24.03\% & 25.92\% & 25.98\% \\
    \hline
     \textbf{w. opt.} & 5.21\% & 7.90\% & 6.17\% & 6.72\% & 7.38\% & 7.05\% & 5.81\% & 4.86\% \\
    \hline
  \end{tabular}
  }\vspace{-0.3cm}
\end{table}

\subsubsection{\textbf{Memory Accesses}}
\label{subsubsec: memory access}
We utilize the \textit{perf} tool to record CPU hardware counters for data reads and cache hits. Table \ref{tab:cache misses} shows the L1 cache misses during index construction before and after our optimization across eight datasets. Each evaluation is conducted on one million samples per dataset, maintaining consistent indexing parameters. The results indicate a consistent reduction in cache misses across all eight datasets with \texttt{Flash}, demonstrating its efficacy in accelerating index construction by minimizing random memory accesses. Given that cache access is approximately 100 times faster than RAM access, even minor reductions in cache misses can significantly enhance the speed of distance computations \cite{ColemanSSS22}.

\subsubsection{\textbf{Arithmetic Operations}}

\begin{table}[t]
\vspace{-0.1cm}
 \setstretch{0.9}
 \fontsize{7.5pt}{4mm}\selectfont
  \caption{Indexing time without and with SIMD optimization.}
  \vspace{-0.4cm}
  \label{tab:SIMD opt}
  \setlength{\tabcolsep}{.0045\linewidth}{
  \begin{tabular}{l|l|l|l|l|l|l|l|l}
    \hline
     & SSNPP & LAION & COHE. & BIGC. & IMAGE. & DATAC. & ANTON & ARGIL. \\
    \hline
    \hline
     \textbf{w/o opt. (s)} & 233 & 154 & 270 & 214 & 88 & 154 & 157 & 140 \\
    \hline
     \textbf{w. opt. (s)} & 129 & 131 & 164 & 147 & 84 & 135 & 138 & 111 \\
    \hline
  \end{tabular}
  }\vspace{-0.2cm}
\end{table}

To assess the impact of SIMD optimization, we compare indexing times with and without this option in Table \ref{tab:SIMD opt}, using the same dataset settings as described in Section \ref{subsubsec: memory access}. The results indicate that SIMD optimization can reduce indexing time by up to 45\%, highlighting its effectiveness in accelerating index construction through efficient arithmetic operations. Notably, the indexing times in Table \ref{tab:SIMD opt} include vector coding time, which constitutes approximately 10\% of the total indexing time; however, SIMD optimization does not affect the coding time.

\begin{table}[t]
 \setstretch{0.9}
 \fontsize{7.5pt}{4mm}\selectfont
  \caption{Coding time (CT) and total indexing time (TIT).}
  \vspace{-0.4cm}
  \label{tab:coding time}
  \setlength{\tabcolsep}{.007\linewidth}{
  \begin{tabular}{l|l|l|l|l|l|l|l|l}
    \hline
     & SSNPP & LAION & COHE. & BIGC. & IMAGE. & DATAC. & ANTON & ARGIL. \\
    \hline
    \hline
     \textbf{CT (h)} & 0.016 & 0.049 & 0.049 & 0.048 & 0.060 & 0.087 & 0.165 & 0.164 \\
    \hline
     \textbf{TIT (h)} & 0.599 & 0.542 & 0.668 & 0.619 & 0.487 & 0.730 & 1.035 & 1.083 \\
    \hline
  \end{tabular}
  }\vspace{-0.3cm}
\end{table}

\begin{figure}
  \setlength{\abovecaptionskip}{0cm}
  \setlength{\belowcaptionskip}{-0.4cm}
  \centering
  \footnotesize
  \stackunder[0.5pt]{\includegraphics[scale=0.33]{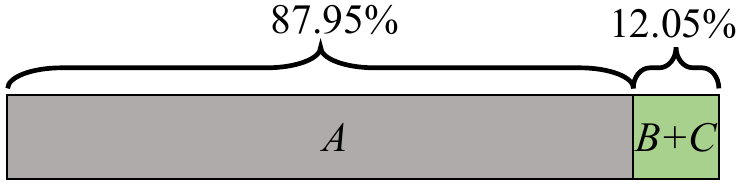}}{(a) LAION-1M ($D=768$)}
  \hspace{0.15cm}
  \stackunder[0.5pt]{\includegraphics[scale=0.32]{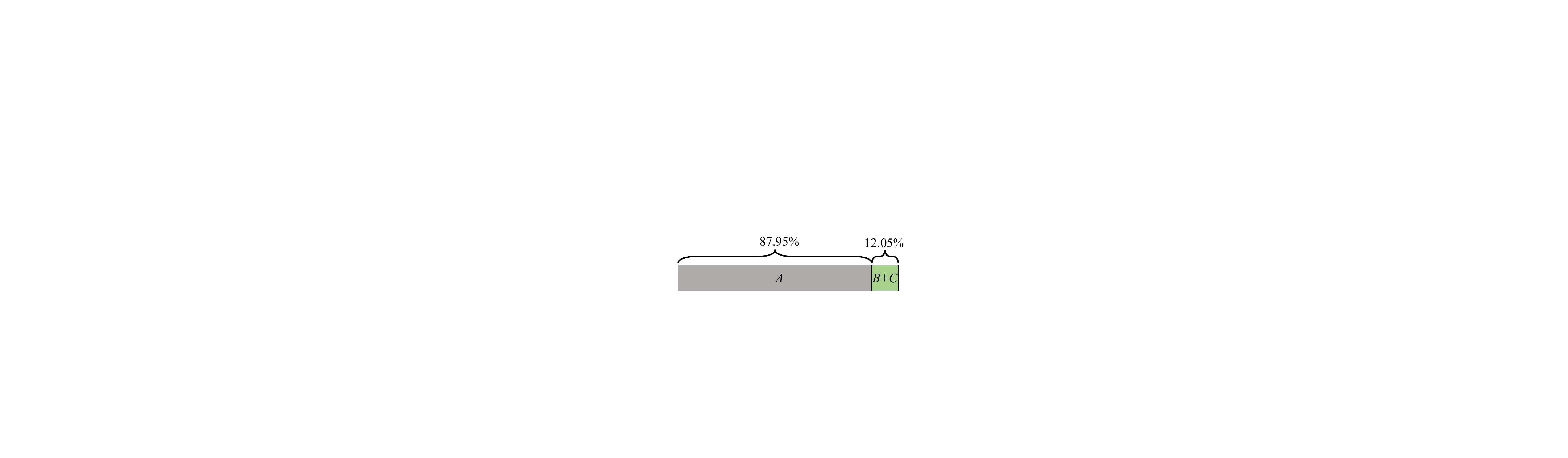}}{(b) ARGILLA-1M ($D=1024$)}
  \newline
  \caption{Profiling of graph index construction time in HNSW-Flash. The distance computation time (\textit{B}+\textit{C}) consists of memory accesses (\textit{B}) and arithmetic operations (\textit{C}). \textit{A} is the time of other tasks, such as data structure maintenance.}
  \label{fig: graph construction time profile}
  \vspace{-0.1cm}
\end{figure}

\subsubsection{\textbf{Profiling of indexing time}}
Table \ref{tab:coding time} presents the coding time (CT) for \texttt{Flash} alongside the total indexing time (TIT) for HNSW-Flash, with TIT encompassing both CT and the graph index construction time (GIT). Data scales for all datasets align with those in Figure \ref{fig: index time}.
The results show that CT constitutes only 10\% of TIT, demonstrating that \texttt{Flash} requires little processing time. We further profile GIT using the \textit{perf} tool in Figure \ref{fig: graph construction time profile}. The results indicate that distance computation occupies a small fraction of GIT for HNSW-Flash. This suggests the primary bottlenecks associated with memory accesses and arithmetic operations in HNSW have been removed through the integration of \texttt{Flash}.

\subsection{Parameter Sensitivity}
\label{subsec: param sensitivity}

\begin{figure}
  \setlength{\abovecaptionskip}{0cm}
  \setlength{\belowcaptionskip}{0cm}
  \centering
  \footnotesize
  \stackunder[0.5pt]{\includegraphics[scale=0.2]{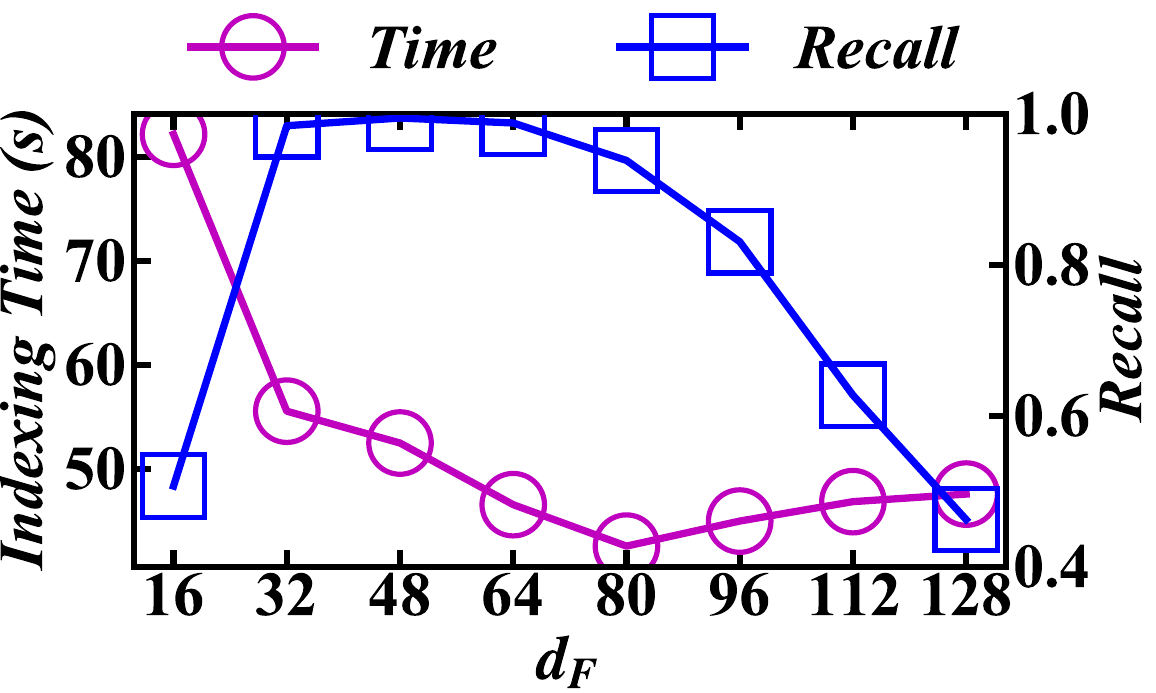}}{(a) LAION-1M ($M_{F}=16$)}
  \hspace{0.15cm}
  \stackunder[0.5pt]{\includegraphics[scale=0.2]{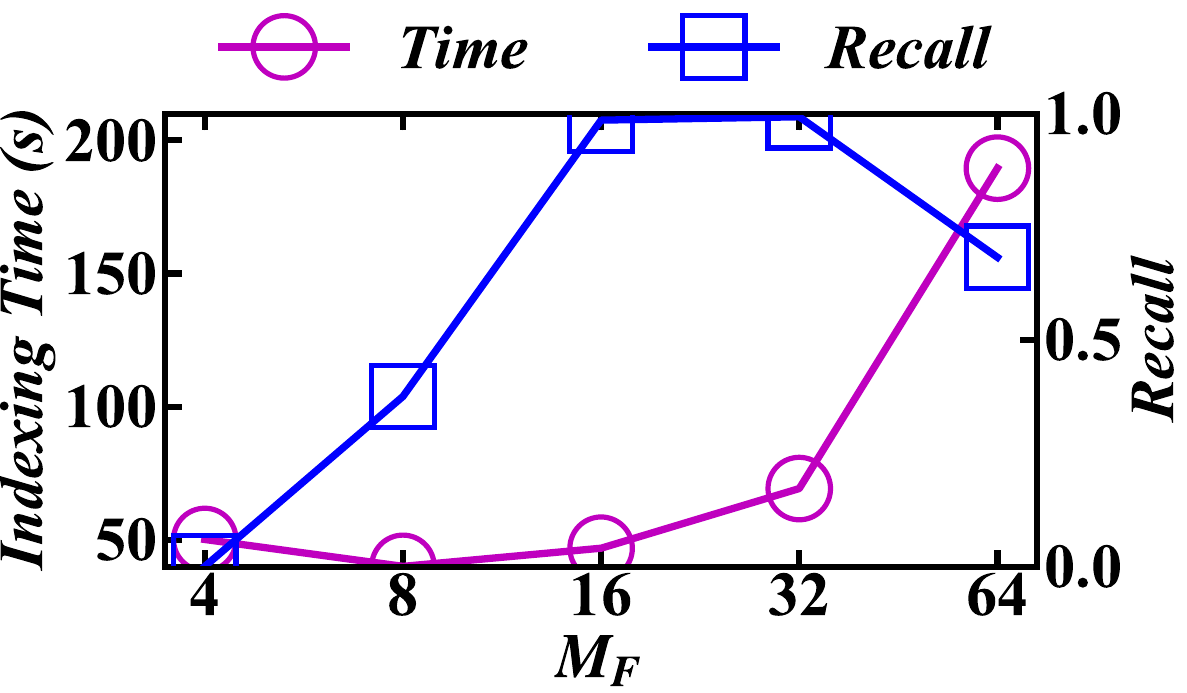}}{(b) LAION-1M ($d_{F}=64$)}
  \newline
  \caption{Effect of parameters within \texttt{Flash}.}
  \label{fig: HNSW-Flash param}
  \vspace{-0.3cm}
\end{figure}

HNSW-Flash has two adjustable parameters to balance construction efficiency and index quality: the dimensions of principal components ($d_F$) and the number of subspaces ($M_F$). Other parameters, like bits per subspace ($L_F$) are fixed to align with index features and hardware constraints.
Figure \ref{fig: HNSW-Flash param} illustrates indexing time and search accuracy across varying $d_F$ and $M_F$, using the recall rate at a given latency as a measure of index quality. In Figure \ref{fig: HNSW-Flash param}(a), with $M_F$ fixed at 16, index quality improves with increasing $d_F$ initially but subsequently declines beyond a certain threshold. Correspondingly, indexing time decreases initially, then rises after reaching a critical point. This phenomenon arises from potential information loss at lower dimensions and redundancy at higher dimensions, given the fixed $L_F$. Notably, the optimal points for indexing time and search accuracy occur at different $d_F$, reflecting the distinct sensitivities of index construction and search procedures to information content. While the search focuses on nearest neighbors, index construction requires navigable neighbors, even if more distant \cite{HNSW,NSG}.
In Figure \ref{fig: HNSW-Flash param}(b), with $d_F$ set to 64, increasing $M_F$ extends indexing time. The search accuracy improves initially but later declines due to increased computational overhead and additional register loads.
\section{Discussion}
\label{sec: summary}
Based on our research, we highlight three notable findings:

(1) \textsf{A compact coding method that significantly enhances search performance may not be suitable for index construction.} Recent research \cite{DiskANN,yue2023routing,HVS,yang2024bridging} has integrated compact coding techniques, such as PQ \cite{PQ}, into the search process while preserving the original HNSW index construction. This approach improves search performance but incurs additional indexing time. Note that the index construction process presents greater challenges due to its more complex execution logic relative to the search procedure \cite{HNSW,graph_survey_vldb2021}. Consequently, directly applying existing compact coding methods to HNSW construction may not yield substantial improvements in indexing speed. Our empirical results indicate that while HNSW-SQ enhances search performance compared to standard HNSW, it provides only marginal improvements in indexing speed.

(2) \textsf{Reducing more dimensions may bring higher accuracy.} While it might seem intuitive that preserving more dimensions during vector dimensionality reduction would increase accuracy \cite{ADSampling,tau-MG,yang2024bridging}, focusing on principal components is advisable when vector representation is constrained by limited bit allocation. This approach reduces the impact of less significant components, which tend to suffer from poor bit usage. For instance, \texttt{Flash} leverages lower-dimensional principal components, significantly accelerating index construction and delivering superior search performance compared to higher-dimensional counterparts (Figure \ref{fig: HNSW-Flash param}).

(3) \textsf{Encoding vectors and distances with a tiny amount of bits to align with hardware constraints may yield substantial benefits.} In compact coding methods, fewer bits typically enhance efficiency but may compromise accuracy, while more bits can improve accuracy at the cost of efficiency \cite{PQ,OPQ}. Therefore, determining an optimal bit count is essential for balancing accuracy and efficiency. However, this optimal count may not align with hardware constraints. Addressing this requires redefining the balance to accommodate hardware specifications. With tailored hardware optimizations, using fewer bits may achieve a better balance than the original optimal counts. For example, \texttt{Flash} encodes distances with 8 bits to leverage SIMD capabilities, achieving a more favorable balance than the standard 32-bit configuration.

\section{Conclusion}
\label{sec: conclusion}
In this paper, we investigate the index construction efficiency of graph-based ANNS methods on modern CPUs, using the representative HNSW index. We identify that low construction efficiency arises from high memory access latency and suboptimal arithmetic operation efficiency in distance computations. We develop three baseline methods incorporating existing compact coding techniques into the HNSW construction process, yielding valuable lessons. This leads to the design of \texttt{Flash}, tailored for the HNSW construction process. \texttt{Flash} reduces random memory accesses and leverages SIMD instructions, improving cache hits and arithmetic operations. Extensive evaluations confirm \texttt{Flash}'s superiority in accelerating graph indexing while maintaining or enhancing search performance. {We also apply \texttt{Flash} to other graph indexes, optimized HNSW implementations, and various SIMD instruction sets to verify its generality.}

\bibliographystyle{ACM-Reference-Format}
\bibliography{myref}

\end{document}